\documentclass{amsart}
\usepackage{latexsym}

\usepackage{amsmath,amssymb,mathrsfs,pictex} 

\usepackage{amsmath,amsthm,amsfonts,amscd,eucal}
\usepackage{graphicx}
\usepackage{epsfig}
\usepackage{epstopdf} % to include  graphics files with pdfLaTeX

\numberwithin{equation}{section}

\def\cb{{\mathcal B}}

\def\cd{{\mathcal D}}

\def\cf{{\mathcal F}}

\def\ch{{\mathcal H}}

\def\car{{\mathcal R}}

\def\cw{{\mathcal W}}

\def\bc{{\mathbb C}}

\def\bn{{\mathbb N}}

\def\br{{\mathbb R}}

\def\bt{{\mathbb T}}

\def\bz{{\mathbb Z}}

\def\ga{{\mathfrak A}} 
\def\gb{{\mathfrak B}}

 \def\gph{{\mathfrak h}}

\def\gw{{\mathfrak W}}

\def\a{\alpha}
\def\b{\beta}
\def\g{\gamma}        \def\G{\Gamma}
\def\d{\delta}        \def\D{\Delta}
\def\eps{\varepsilon}

\def\th{\vartheta}

\def\l{\lambda}       \def\La{\Lambda}
\def\m{\mu}
\def\n{\nu}

\def\r{\rho} \def\R{\mathop{\rm P}}
\def\s{\sigma}        \def\S{\Sigma}
     
\def\t{\tau}

\def\f{\varphi} \def\F{\Phi}

\def\om{\omega}        \def\Om{\Omega}

\newtheorem{Thm}{Theorem}[section]

\newtheorem{Prop}[Thm]{Proposition}
\newtheorem{Lemma}[Thm]{Lemma}

\theoremstyle{definition}
\newtheorem{Dfn}[Thm]{Definition}

\newtheorem{Rem}[Thm]{Remark} 
\theoremstyle{remark}
\def\di{\mathop{\textrm d}\!}

\def\im{\mathop{\rm Im}}

\def\spr{\mathop{\rm spr}}

\newcommand{\ssv}{\text{span}\,}

\newcommand{\tr}{\textrm{Tr}}
\newcommand{\comb}{\dashv}

\def\idd{{\bf 1}\!\!{\rm I}}
\def\spr{\mathop{\rm spr}}

\def\supp{\mathop{\rm supp}}

\newcommand{\nn}{\nonumber}
\begin{document}

%titolo
\title[harmonic analysis and Bose--Einstein condensation]
{harmonic analysis on inhomogeneous amenable networks and the Bose--Einstein condensation}
%Autori
\author{Francesco Fidaleo}
\address{Dipartimento di Matematica,
Universit\`{a} di Roma Tor Vergata, 
Via della Ricerca Scientifica 1, Roma 00133, Italy} 

\email{fidaleo@mat.uniroma2.it}


\date{\today}

\begin{abstract}
We study in detail relevant spectral properties of the adjacency matrix of inhomogeneous amenable networks, and in particular those arising by negligible additive perturbations of periodic lattices. The obtained results are deeply connected to the systematic investigation of the 
Bose--Einstein condensation for the so called Pure Hopping model describing the thermodynamics of Bardeen--Cooper pairs of Bosons in arrays of Josephson junctions. After a careful investigation of the infinite volume limits of the finite volume adjacency matrix corresponding to the (opposite of the) Hamiltonian of the system, the main results can be summarised as follows. First, the appearance of the Hidden Spectrum for the Integrated Density of the States in the region close to the bottom of the Hamiltonian, implies that the critical density is always finite. Second, we show that the Bose--Einstein condensation can appear if and only if the adjacency matrix is transient, and not just when the critical density is finite. We can then exhibit examples of networks for which condensation effects can appear in a natural way even if the critical density is infinite and vice--versa, that is when the critical density is finite but the system does not admit any locally normal state exhibiting condensation. Contrarily to the known homogeneous examples, we also exhibit networks whose geometrical dimension is less than 3, for which the condensation takes place. Due to non homogeneity, particles may condensate even in configuration space. The shape of the ground state wave--function (i.e. the Perron--Frobenius weight) provides the spatial distribution of the condensate. Such a spatial distribution of the condensate is described by the Perron--Frobenius dimension defined in a natural way. For systems for which the critical density is finite and the adjacency matrix is transient, we show that, if the Perron--Frobenius dimension is greater that the geometrical one, we can have condensation only if the mean density of the state is infinite. Conversely, in the opposite situation when the geometrical dimension exceeds the Perron--Frobenius one, the condensation appears only for states with mean density coinciding with the critical one, that is the amount of the condensate is negligible with respect to the amount of the whole particles. All those states are KMS ones with respect to the natural dynamics generated by the formal Pure Hopping Hamiltonian. The existence of such a dynamics, which is a delicate issue, is provided in detail
\end{abstract}

\maketitle

\tableofcontents

\section{Introduction}

The present paper is devoted to the detailed analysis of some relevant spectral properties of the adjacency matrix (simply denoted by the Adjacency) of a wide class of amenable inhomogeneous graphs. The obtained results are applied to the systematic investigation of the  
Bose--Einstein Condensation (BEC for short) for the so--called Pure Hopping model. These spectral results of mathematical nature have also a self--containing interest.

The investigation of the BEC has a very long history after discovering of a new statistics by Satyendra Nath Bose and Albert Einstein at the beginning of 20th century. It concerns elementary particles obeying the Bose--Einstein statistics, that is having integer spin. Roughly speaking, it means that a macroscopic amount of particles can occupy the ground state after the thermodynamic limit. We mention the ideal model consisting of a gas of free massive Bosons (see the volume II of \cite{BR2} for a detailed analysis and for the huge literature cited therein). Also many phenomena involving quasi--particles like phonons or magnons can be described by the BEC. Recently, in \cite{KSVW} the condensation of massless particles like photons has been pointed out. The reader is referred to the very huge existing literature for details on the BEC. Also the phenomenon of the superfluidity of the helium isotope $He_4$, whose explanation is due to Lev Landau (cf. \cite{L}), seems to be tightly but not directly connected with the BEC. It is well known that Fermions (i.e. quantum particles of half--integer spin) does not lead to any condensation by the Pauli exclusion principle. Nevertheless, the other isotope helium $He_3$ still exhibits superfluidity at a temperature very close to $0^o$ Kelvin, even if these are Fermi particles. The superfluidity of $He_3$ can be justified with the fact that at very low temperature, $He_3$--particles form the so--called sea of pair--particles which can be considered as Bosons, that is also pairs of $He_3$ particles can enjoy the properties of a Bose fluid. According to the BCS theory (cf. \cite{BC}), it is precisely the same phenomenon occurring in superconductors where pairs of electrons forms the so called Bardeen--Cooper Bosons. Thus, the phenomenon of superconductivity is also connected with the BEC.

Concerning the superconductivity, recently in the paper \cite{BCRSV}, it was shown the surprising fact that the critical density describing the condensation of Bardeen--Cooper Bosons for the Pure Hopping model can be finite also for low dimensional networks like the Comb and the Star graph.

The model describes arrays of Josephson junctions in a sea of Bardeen--Cooper pairs. It might be described after some reasonable approximations by the so--called {\it Pure Hopping model} whose multi--particle Hamiltonian is given by
\begin{equation}
\label{boha1}
H_{PH}=-J\sum_{i,j}A_{i,j}a^{\dagger}_ia_j\,.
\end{equation}
Here, $J$ is a coupling constant to be determined experimentally, and $A$ is the Adjacency of the network under consideration. On the other hand, some promising experiments (cf. \cite{Lo} and references cited therein) on the 
the Comb and Star Graphs have been done, pointing out an enhanced current at low
temperature which could be explained by condensation phenomena.

As the Hamiltonian \eqref{boha1} is quadratic, it can be diagonalised. This simply means (cf. \cite{BR2}) that one can directly reduce the matter to the one--particle Hamiltonian which, after putting $J=1$ and normalising such that it is positive with $0$ as the ground state energy level, is simply written as
 \begin{equation}
\label{boha2}
 H=\|A\|\idd-A\,.
 \end{equation}
The Hamiltonian \eqref{boha2} leads to the investigation of a pure topological model on the graph whose main ingredient is the Adjacency. 

The study of the BEC is strongly connected with the investigation of relevant spectral properties of Hamiltonian. The particular form of the Pure Hopping one--particle Hamiltonian reduces the matter to the investigation of spectral properties of well known mathematical objects which have a self--containing interest. The bridge between mathematics and physics can be easily explained as follows. Consider the Bose--Gibbs occupation number (cf. \cite{LL})
at the energy $\eps$, inverse temperature $\b>0$, and chemical potential $\m<0$, given by
 \begin{equation}
 \label{bgokn}
n(\eps)=\frac1{e^{\b(\eps-\m)}-1}\,.
 \end{equation}
It naturally involves the operator $\big(e^{\b(H-\m \idd)}-\idd\big)^{-1}$ acting on the one--particle Hilbert space, that 
for low energies of the Hamiltonian \eqref{boha2}, and after using Taylor expansion, one heuristically gets
\begin{equation}
\label{euro}
\frac1{e^{\b(H-\m \idd)}-\idd}\approx\frac1{\b(H-\m\idd)}=\frac1{\b((\|A\|-\m)\idd-A)}
\equiv \frac1{\b}R_A(\|A\|-\m)\,.
\end{equation}
It is well known (see e.g. Section 5.1 of \cite{BR2}) that the BEC is connected with the spectral properties of the Hamiltonian for values of the energies close to ground one. Then for the Pure Hopping model, the study of the BEC is reduced to the investigation of the spectral properties of the Resolvent
$R_A(\l)$, for $\l\approx\|A\|$, the latter being a very familiar object for mathematicians. 

After reducing the matter to the Adjacency by using \eqref{euro}, it emerges for various amenable and not amenable models, that the condensation phenomena are deeply related to the results of mathematical nature listed below.
\begin{itemize}
\item[(i)] The part of the spectrum close to the bottom of the one--particle Hamiltonian, corresponding to the part of the spectrum near the norm of the Adjacency, might not contribute to the integrated density of the states. This property is called the (appearance of the) Hidden Spectrum. 
\end{itemize} 
Hidden Spectrum automatically leads to the finiteness of the critical density of the model.
We show that for most of the model under consideration, it is possible to determine whether the Hidden Spectrum appears by solving a unique equation (the Secular Equation) whose unknown is the norm of the Adjacency of the perturbed graph.
\begin{itemize}
\item[(ii)] The existence of locally normal states exhibiting BEC is determined by the transiency/recurrence character of the Adjacency, and not just by the finiteness of the critical density.
\end{itemize}
The locally normal states (i.e. those whose the local density of the particles is finite) exhibiting BEC are explicitly constructed when the Adjacency is transient, by fixing the portion of the condensate.
We then can exhibit examples of networks for which the critical density is infinite but exhibiting BEC, and vice--versa.
\begin{itemize}
\item[(iii)]
The sequence of the Perron--Frobenius (PF for short) eigenvalues suitably normalised, of the Adjacency of the finite volume theories, 
converges to a unique PF weight for the  Adjacency, which we can describe explicitly. 
\end{itemize}
The "shape", quantitatively described by the so called PF dimension, takes into account of the distribution 
of the condensate on the space of configurations. It is not uniformly distributed on the graph due to inhomogeneity. 

As a consequence of all these results, we also prove that if the Adjacency of the graph is recurrent (transient and the geometrical dimension exceeds the PF one), no locally normal states exhibiting condensation can be constructed at all (can be constructed for a mean density of the particle greater than the critical one, respectively). We also establish another unexpected fact: it is possible to exhibit locally normal states describing condensation, even if the geometrical dimension is less then 3. The unexpected emerging results
are summarised  in the following table.

\vskip1cm

\begin{tabular}{||ll|c|c|c|c|c|c|c||}
\hline
&& \!\!\!  $\r_c$\!\!\!  & \ \ \!\!\!\!\!\!\!  R/T\!\!\!\!\!   \ \  & \ \ \!\!\! $d_G$\!\!\!  \ \ & \ \ \!\!\!\!\!\!  $d_{PF}$\!\!\!\!\!\!   \ \ & \ \ \!\!\!  0--BEC\!\!\!   \ \ &  \ \ \!\!\!\!\!\!  $\r$--BEC \!\!\! \!\!\!   \ \ & \ \ \!\!\! \!\!\! $\infty$--BEC\!\!\! \!\!\!   \ \ \\
\hline

${\bz}^d, d<3$ &&$\infty$  &R &$d$&$d$&no&no&no\\ \hline 			
${\bz}^d, d\geq3$ &&$<\infty$  &T &$d$&$d$&no&yes&no\\ \hline 
star graph&&$<\infty$  &R &$1$&$0$&no&no&no\\ \hline 
$\bz^d\dashv\bz, d<3$\!\!\!\!\!\!  &&$<\infty$  &R &$d+1$&$d$&no&no&no\\ \hline 
$\bz^d\dashv\bz, d\geq3$\!\!\!\!\!\!  &&$<\infty$  &T &$d+1$&$d$&yes&no&no\\ \hline 
$\bn$&&$\infty$  &T &$1$&$3$&no&no&yes\\ \hline 
$\bn\dashv\bz$&&$<\infty$  &T &$2$&$3$&no&no&yes\\ \hline
$\bn\dashv\bz^2$&&$<\infty$  &T &$3$&$3$&no&yes&no\\ \hline 
\end{tabular}

\vskip1cm
Here, $G\dashv H$ is the comb shaped graph whose base graph is $G$, $\r_c$ is the critical density, R/T denotes the recurrence/transience character of the Adjacency, 0--BEC, $\r$--BEC and $\infty$--BEC 
denote the existence of locally normal states 
exhibiting BEC at mean density $\r=\r_c$,   $\r\in(\r_c,+\infty)$, and finally 
$\r=+\infty$, respectively.

The new and very surprising phenomena can better explained as follows. First, Hidden Spectrum always implies
the finiteness of the critical density.
This leads to the fact that BEC can appear also in low dimensional cases $d<3$. In addition, not for all the models with finite critical density, it is possible to exhibit locally normal states describing BEC. It depends on the transience/recurrence character of the Adjacency: the model can exhibit states with BEC if and only it is transient, even if the critical density is infinite. On the other hand, there are recurrent models for which the critical density is finite without exhibiting any locally normal states describing condensation.

Another new aspect is the difference between the geometrical dimension of the network and the growth of the wave function of the ground state. The last fact can be explained for amenable cases by introducing the Perron--Frobenius dimension. 
Fix a graph $G$ equipped with an exhaustion $\{\La_n\}_{n\in\bn}$ together with the PF weight $v$ for the adjacency 
$A$ obtained as infinite volume limit of the sequence of the finite volume PF eigenvectors normalised at 1 on a fixed root.\footnote{By passing to a subsequence, it is easily shown that the sequence of finite volume PF eigenvectors as above, converges by compactness to a PF weight, see Proposition 4.1 of \cite{FGI1}. In all the situations considered in the present paper, the sequence of finite volume PF eigenvectors for the chosen exhaustion converges point--wise to a PF weight without passing to subsequences, as we will see below.}
The geometrical dimension $d_{G}(G)$ of the network $G$ is defined to be $a$ if
$|\La_{n}|\sim n^{a}$. The Perron--Frobenius dimension
$d_{PF}(G)$ of $G$ is defined to be $b$ if 
$\left\|v\lceil_{\ell^2(\La_n)}\right\|\sim n^{b/2}$. For homogeneous cases (i.e. when the valence is constant), we have $d_G=d_{PF}$ for most the relevant models.
The possible difference between the geometry of the graph $G$
and the 
$\ell^2$--behavior of the norm $\left\|v\lceil_\La\right\|^2$ of the ground state wave function $v$, encoded in the geometrical and PF dimensions, respectively, gives rise to the following facts in the infinite volume limit $\La\uparrow G$. Fix a network $G$ with $A_G$ transient and $\r_c<+\infty$, and consider the Pure Hopping model on it.
If we start by fixing a priori the amount of the condensate 
by a careful choice of the sequence of the finite volume chemical potentials $\m(\La_n)\to0$, we obtain for the spatial density of such an amount of the condensate:
$$
C(\La)\approx D\frac{\left\|v\lceil_\La\right\|^2}{|\La|}\,,
$$
where $D>0$ is fixed. If 
$d_{PF}>d_G$ then $C(\La)\to+\infty$ which implies $\r(\om)=+\infty$, and if $d_{PF}<d_G$ then $C(\La)\to0$ which leads to $\r(\om)=\r_c$. Conversely, if we fix the density
$\r>\r_c$, we show that the amount of the condensate is now given by
$$
C(\La)\approx(\r-\r_c)\frac{|\La|}{\left\|v\lceil_\La\right\|^2}\,.
$$
The latter means that, if $d_G>d_{PF}$ then $C(\La)$ would become infinite and this has as a consequence that the diagonal part of the two--point function of the quasi--free state diverges. Namely, no locally normal states exhibiting BEC can be constructed as infinite volume limit of finite volume states with the constrain $\r>\r_c$. If conversely 
$d_G<d_{PF}$, then $C(\La)\to0$, giving for the limiting density again $\r(\om)=\r_c$, even if the finite volume densities are all constant, bigger than the critical one by construction.
Due to non homogeneity, particles condensate also in the configuration space. Thus, the above considerations are nothing but the naive explanation of the fact that, after cooling below the critical temperature, the  
system undergoes a "dimensional transition" (proven in details below) governed by the possible difference between the geometrical and PF dimensions, naturally appearing in inhomogeneous cases. 

Some of the results listed above are proved in \cite{FGI1} for a class of amenable comb graphs, and in
\cite{F} for perturbations of Cayley Trees, that is in non amenable situations where the boundary effects cannot be neglected in the infinite volume limit. The present paper is mainly devoted to prove systematically, the  relevant spectral results relative to the Adjacency listed above, and then to investigate in the full generality the thermodynamics of the Pure Hopping models for the networks listed in the table above, relatively to the non trivial cases $\bn$, $\bn\dashv\bz^d$, and $\bz^d\dashv\bz$.\footnote{The homogeneous cases $\bz^d$ is extensively treated in literature, whereas the recurrent examples for which the PF weight is normalisable (i.e. a PF eigenvector, necessarily unique up a phase factor) are trivial.}
For those models, we explicitly prove the existence of locally normal states exhibiting BEC, provided that the Adjacency is transient, independently on the finiteness of the critical density. In addition, we prove that such states satisfy the Kubo--Martin--Schwinger (KMS for short) boundary condition with respect to the natural dynamics associated to the
formal Pure Hopping Hamiltonian \eqref{boha1}. The existence of such a dynamics is a delicate issue which is provided in detail. 

The present paper is organised as follows. After recalling the standard definitions and the main properties of a graph and its Adjacency, Section \ref{sec:prel} collects some  relevant results relative to zero--density perturbation graphs. Among those, we list the formulae for the perturbed Adjacency, known as the Krein Formula, and for
the Laplace transform of its Integrated Density of the States. For the reader's convenience, we include a section (cf. Section \ref{sec:stmec}) devoted to some results involving the statistical mechanics on graphs. Those include some of general nature, and others concerning the Pure Hopping models and its particularisation to density--zero perturbations. Section \ref{sec:gen} concerns some general facts relative to the Pure Hopping model. Among them we mention those relative to the PF dimension, and the Secular Equation allowing to compute the norm of the perturbed Adjacency in order to decide whether the Hidden Spectrum appears. The main result of general nature is that the Pure Hopping model can exhibit BEC if and only if the Adjacency is transient. It is also show that the formal Hamiltonian \eqref{boha1} generates a dynamics on a Canonical Commutation Relations $C^*$--algebra containing all the Weyl unitaries $\{W(\d_x)\mid x\in G\}$ and globally stable for the time evolution, such that all such states exhibiting BEC satisfy the KMS boundary condition. Section \ref{sec:comb} deals with the so called comb product networks for recurrent and transient situations. We also compute the PF weights and the corresponding PF dimensions, even for situations not exhibiting Hidden Spectrum, the last being much more complicate to manage. As an intermediate result, we compute the asymptotic of the finite volume $\ell^2$--norms of 
$$
f({\bf k})=\langle R_{A_{\bz^d}}(2d)\d_{\bf k},\d_{\bf0}\rangle\,,
$$ 
for the non trivial transient cases $d=3,4$. This Tauberian result, probably known to the experts, may have an interest in itself. Finally, Sections \ref{sec:en}, \ref{snd},
\ref{sec:ze} deal with the graph $\bn$, and the comb graphs $\bn\comb\bz^d$, $\bz^d\comb\bz$, respectively. We investigate in details the needed spectral properties of the corresponding Adjacencies, and apply the results to the BEC. For sequences $\{\m(\La_n)\mid\La_n\subset G\}$ of finite volume chemical potentials, we cover all the situations corresponding the condensation regime $\lim_n\m(\La_n)=0$, including that corresponding to fixing the amount of the condensate, and the usual one (not suitable for the inhomogeneous networks under consideration) obtained by fixing the mean density $\r\geq\r_c$ of the system. The case $\r<\r_c$, being straightforward (cf. Section 5.2.5 of \cite{BR2}), is left to the reader.

\section{Prelimiaries}
\label{sec:prel}

A {\it graph} (called also a {\it network}) $X=(VX,EX)$ is a collection $VX$ of objects,
 called {\it vertices} and denoted by $x\in VX$, together with a collection $EX$ of unordered lines connecting vertices, 
 called {\it edges} and denoted by $e_{xy}$. Multiple edges, as well as self--loops are allowed. 
 Let $E_{xy}:=\{e_{xy}\mid x\sim y\}$ be the collection of all the edges connecting $x$ with $y$. We have $E_{xy}=E_{yx}$. The 
 {\it degree} of $x\in X$ is defined as
$$
\deg(x):=|\{E_{xy}\mid y\in VX\}|\,.
$$
Let us denote by $A=[A_{xy}]_{x,y\in X}$, $x,y\in VX$, the {\it adjacency
 matrix} of $X$, called simply {\it Adjacency}, and given by
 $$
 A_{xy}:=|E_{xy}|\,.
 $$ 
All the geometric properties of $X$ are encoded in $A$. For example, $X$ is connected if and only if $A$ is irreducible. Setting
 $$
 \deg:=\sup_{x\in VX} \deg(x)\,,\quad D_{xy}:=\deg(x)\d_{x,y}\,,
 $$
we have
 $\sqrt{\deg}\leq\|A\|\leq \deg$, that is $A$ is bounded if and only if $X$
has uniformly bounded degree. Denoting $D=[D_{xy}]_{x,y\in X}$ the {\it degree matrix}, 
the discrete {\it Laplacian} is defined as
$$
\D:=A-D\,,
$$
with the convention $\D\leq0$ standardly used in physical literature. Then the Pure Hopping Hamiltonian can be viewed as a 
Schr\"odinger operator $H=-\D+V$, where $V$ is the multiplication operator on $\ell^2(VX)$ for the function
$$
V(x)=\|A\|-\deg(x)\geq0\,.
$$
The Laplacian considerably differs from the Adjacency for inhomogeneous networks, as those considered in the present paper.

For connected networks (or on each connected component) we can define the standard distance
\begin{equation*}
d(x,y):=\{\min\ell(\pi(x,y))\mid\,\pi(x,y)\,\text{path connecting}\,x,y\}\,,
\end{equation*}
$\ell(\pi)$ being the length of the path $\pi$ (i.e. the number of the edges in $\pi$).
In the present paper, all the graphs are connected, countable and with uniformly bounded
degree. In addition, we deal only with bounded operators acting on $\ell^2(VX)$ if it is not otherwise specified.

Let $B$ be a closed (not necessarily bounded) operator acting on a Banach space, and 
$\l\in{\rm P(B)}\subset\bc$ the resolvent set of $B$. As usual,
$$
R_B(\l):=(\l\idd-B)^{-1}
$$
denotes the {\it Resolvent} of $B$.

Fix a bounded matrix with positive entries $B$ acting on $\ell^2(VX)$. Such an operator is called {\it positive preserving} as it preserves the elements of $\ell^2(VX)$ with positive entries. If $B$ is self--adjoint, $B$ is {\it positive} if $\langle Bu,u\rangle\geq0$ for each $u\in\ell^2(VX)$. Examples of self--adjoint operators which are positive--preserving but not positive and vice--versa are the Adjacency and the (opposite of the) Laplacian, respectively.

Fix a positive preserving operator $B$ acting on $\ell^2(VX)$. The sequence $\{v(x)\}_{x\in VX}$ is called a (generalized) {\it Perron--Frobenius eigenvector} if it has positive entries and
 $$
 \sum_{y\in VX}B_{xy}v(y)=\spr(B)v(x)\,,\quad x\in VX\,.
 $$
where "spr" stands for spectral radius. 
If such a vector is normalizable (i.e. if it belongs to $\ell^2(VX)$) it is a standard $\ell^2(VX)$--vector, otherwise it is a weight, simply denoted as a PF weight.
\begin{Dfn}
For a graph $X$, denote by $\cw_X$ the set of the 
PF weight $w$ for $A_X$ which can be obtained as point--wise limit of some sequence of PF eigenvectors $\{w_n\}_{n\in\bn}$ associated to some exhaustion $\{X_n\}_{n\in\bn}$ of $X$, normalised to $1$ on a common root 
$o\in X_n$, $n\in\bn$.
\end{Dfn}
Suppose for simplicity that $B$ is irreducible and self--adjoint. It is said to be {\it recurrent} if
\begin{equation}
\label{caz}
\lim_{\l\downarrow\|B\|}\langle R_B(\l)\d_x,\d_x\rangle=+\infty\,.
\end{equation}
Otherwise $B$ is said to be {\it transient}.
It is shown in \cite{S}, Section 6, that the recurrence/transience character of $B$ does not depend on the base--point chosen in computing the limit in \eqref{caz}.
The PF eigenvector is unique up to a multiplicative constant, if $X$ is finite or when $B$ is recurrent, see e.g. \cite{S}. In general, it is not unique, see e.g. Section 3 of \cite{F2} for an example.

We say that an operator $B$ acting on $\ell^2(VX)$ has {\it finite
propagation} if there exists a constant $r=r(B)>0$ such that, for any
$x\in X$, the support of $B\d_x$ is contained in the ball
$$
B(x,r):=\{y\in VG\mid d(x,y)\leq r\}\,,
$$ 
centered in $x$ and with radius $r$. It is easy to show that for the adjacency operator $A$ on $X$, then $A^k$ has propagation $k$ for any integer $k\geq0$.

Let $X$ be an infinitely extended graph with an exhaustion $\{\La_n\}_{n\in\bn}$ which is kept fixed during the analysis. Denote
$P_n$ the orthogonal projection in $\cb(\ell^2(VX)$ associated to the finite region $V\La_n$. We report
the definition of the 
{\it the integrated density of the states} 
of a bounded self--adjoint operator $B\in\cb(\ell^2(VX))$
given in \cite{F1}.
Indeed, consider on 
$\cb(\ell^2(VX))$ the state
$$
\t_n:=\frac1{|V\La_n|}\tr_n(P_n\,{\bf\cdot}\,P_n)\,.
$$
Define for a bounded self--adjoint operator $B$,
\begin{equation}
\label{3434}
\t^B(f(B)):=\lim_n\t_n(f(P_nBP_n))\,,\quad f\in C(\s(B))\,,
\end{equation}
provided such a limit exists. The domain $\cd_{\t^B}\subset C^*(B)$ is precisely the linear subspace of the unital $C^*$--algebra 
$C^*(B)\subset\cb(\ell^2(VX))$ generated by $B$, for which
the limit in \eqref{3434} exists. Notice that the definition of $\t^B$ might depend on the chosen exhaustion $\{\La_n\}_{n\in\bn}$. As the exhaustion is always kept fixed, we omit to indicate such a dependence.
Suppose now that $\cd_{\t^B}=C^*(B)$. Then it provides a state on $C^*(B)$ and, by the Riesz--Markov Theorem, a  Borel probability measure 
$\m_B$ on $\s(B)$, the spectrum of $B$. Thus, there exists a unique right continuous, increasing, positive function $F_B$ satisfying
$$
F_B(x)=0\,,\,\,\,x<\min\s(B)\,;\quad F_B(x)=1\,,\,\,\,x\geq\max\s(B)\,,
$$
such that
$$
\m_B((-\infty,x])=F_B(x)\,,\quad x\in\br\,.
$$
The previous described cumulative function $F_B$ is precisely the integrated density of the states (IDS for short) associated to $B$, provided it exists for the chosen exhaustion. When the graph is amenable and the operator $B$ has finite propagation, the definition and some of the main facts relative to the IDS considerably simplify as the boundary effects play no role in the infinite volume limit, see Theorem 2.1 of \cite{F1}.

Consider the graph $Y$ such that $VY=VX$, both equipped with exhaustions
$\{X_n\}_{n\in\bn}$, $\{Y_n\}_{n\in\bn}$ such that $VY_n=VX_n$,
$n\in\bn$.  If $|\l|$ is sufficiently large, then it is possible to express the Resolvent of $Y$ in terms of the Resolvent of $X$ by the 
{\it Krein Resolvent Formula}. Indeed, fix $X$ as the reference graph and put $A_Y:=A_X+D$,
where $D$ is the perturbation, which acts
on $\overline{\car(D)}$. Put, for $\l\in\bc$,
\begin{equation} 
 \label{s2s}
S(\l):= DPR_{A_X}(\l)\lceil_{\overline{\car(D)}}
\end{equation}
acting on $\overline{\car(D)}$, where $P:=P_{\overline{\car(D)}}$ is the orthogonal projection onto 
$\overline{\car(D)}\subset \ell^2(VX)$. The Krein formula assumes the form
 \begin{equation} 
 \label{k2s}
R_{A_Y}(\l)=R_{A_X}(\l)+R_{A_X}(\l)(P-S(\l))^{-1}DPR_{A_X}(\l)\,.
\end{equation}
By using Neumann expansion in \eqref{k2s}, we can see that the resolvent formula \eqref{k2s} 
holds true for $|\l|>\|A_X\|+\|D\|$. From now on, we suppose that $Y$ is an additive perturbation of $X$ involving only edges, which is the case under consideration in the present paper. We obviously have $\|A^Y\|\geq\|A^X\|$. The formula \eqref{k2s} extends to any simply connected subset containing the point at infinity of $\bc\cup\{\infty\}$ made of the $\l$ for which $\idd_{P\ell^2(VX)}-S(\l)$ is invertible, the last being a subset of $\text{P}(A)$ by the definition \eqref{s2s} of $S(\l)$. Conversely, 
the norm $\|A_Y\|$ of the perturbed graph $Y$ 
might be computed as
$$
\|A_Y\|=\max\{|\l|\leq\|A_X\|\mid\idd_{P\ell^2(VX)}-S(\l)\,\,\text{is not invertible}\,\}\,.
$$
The graph $Y$ is said to be a {\it negligible} or {\it density zero
perturbation} of $X$ if it differs from $X$ by a number of edges
for which
$$
 \lim_n\frac{|\{e_{xy}\in EX\triangle EY\mid x\in VX_n\}|}{|VX_n|}=0\,,
$$
where $EX\triangle EY$ denotes the symmetric difference. 
Suppose $\cd(\t^{A_X})=C^*(A_X)$, that is $A_X$ admits the IDS. Define
\begin{equation}
\label{ddee}
\d:=\|A_X\|-\|A_Y\|\,.
\end{equation}
Denote
$F_X:=N_{\|A_X\|\idd-A_X}$, $F_Y:=N_{\|A_Y\|\idd-A_Y}$, where $N_B$ is the IDS of the operator $B$. 
The following theorem collects Proposition 1.3 of \cite{F1} and Corollary 2.6 of \cite{F}, which we report for the convenience of the reader.
\begin{Thm} 
\label{density0} 
Let $Y$ be an additive negligible perturbation of the network $X$. Then $\cd(\t^{A_Y})=C^*(A_Y)$, and
$\t^{A_Y}(f(A_Y))=\t^{A_X}(f(A_X))$.
In addition,
\begin{equation*}
F_Y(x)=F_X(x+\d)\,.
\end{equation*}
\end{Thm}
Suppose that $\|A^Y\|>\|A^X\|$. Then $\d<0$ in \eqref{ddee}. 
Thus,  the part of the spectrum $[\|A_X\|,\|A_Y\|)$ does not contribute to the IDS of $A_Y$. In this case we say that the perturbed Pure Hopping Hamiltonian $\|A_Y\|\idd-A_Y$ has Hidden Spectrum, see Definition \ref{hyspe} below.
The appearance of the Hidden Spectrum is then the combination of two different effects: the perturbation should be of density zero in order not to affect the IDS (Proposition \ref{density0}), but it should be sufficiently large in order to increase the norm of the perturbed Adjacency. Even in non amenable cases where the growth is exponential and the boundary effects cannot be neglected, it might be sufficient a finite perturbation to have Hidden spectrum, see Section 3 of \cite{F}.

\section{Statistical mechanics on infinitely extended networks}
\label{sec:stmec}

Let $(\ga,\a)$ be a dynamical system consisting of a non Abelian $C^*$--algebra, and a one--parameter group of $*$--automorphism $\{\a_{t}\}_{t\in\br}$.
The  state $\f$ on the $C^{*}$--algebra 
$\ga$ satisfies the Kubo--Martin--Schwinger (KMS for short)
boundary condition at inverse temperature $\b\in\br\backslash\{0\}$ w.r.t the group of
automorphisms $\a$ if
\begin{itemize}
\item[(i)] for every $A,B\in\gb$, $t\mapsto\f(A\a_{t}(B))$, $t\mapsto\f(\a_{t}(A)B)$ are continuous;
\item[(ii)] for each $f\in\widehat{\cd}$,
\begin{equation*}
\int\f(A\a_{t}(B))f(t)\di t=\int\f(\a_{t}(B)A)f(t+i\b)\di t\,,
\end{equation*}
where  '' $\widehat{}$ '' stands for the Fourier transform,
 and $\cd$ is the space  of the infinitely often differentiable, compactly supported functions on $\br$.
\end{itemize}
The following facts are well known. First, a KMS is automatically invariant w.r.t. the automorphism group $\a_t$. Second, the cyclic vector 
$\Om_{\f}$ of the Gelfand--Naimark--Segal (GNS for short) 
quadruple $\big(\pi_{\f},\ch_\f,U_\f,\Om_\f\big)$ is also separating for
$\pi_{\f}(\gb)''$. Denote $\s^{\f}$ its modular group (cf. \cite{BR2}).
According to the definition of KMS boundary condition, we have
\begin{equation*}
\s^{\f}_{t}\circ\pi_{\f}=\pi_{\f}\circ\a_{-\b t}\,.
\end{equation*}
We refer the reader to \cite{BR2} and the literature cited therein, for various equivalent formulations of the KMS condition, proofs, details, and applications. 

The $C^*$--algebras considered here are those arising from the {\it Canonical Commutation Relations} (CCR for short). Namely, let $\gph\subset\bar\gph$ be a subspace of the  Hilbert space $\bar\gph$, equipped with the non degenerate inner product 
$\langle\,{\bf\cdot}\,,\,{\bf\cdot}\,\rangle$, supposed to be linear w.r.t. the first variable. Consider the following (formal) relations,
\begin{equation}
\label{cccrr}
a(f)a^{\dagger}(g)-a^{\dagger}(g)a(f)=\langle g,f\rangle\,\quad f,g\in\gph\,.
\end{equation}
It is well known that the relations \eqref{cccrr} cannot be realised by bounded operators acting on any Hilbert space. A standard way to realise the CCR is to look at the {\it symmetric Fock space} $\cf_+(\bar\gph)$ on which the annihilators $a(f)$ and creators $a^{\dagger}(f)$ naturally act as unbounded closed, adjoint each other (i.e. $a(f)^*=a^{\dagger}(f)$) operators. This concrete representation of the CCR is called the {\it Fock representation}. 
An equivalent description for the CCR is to put on $\cf_+(\bar\gph)$,
\begin{equation*}
\Phi(f):=\overline{\frac{a(f)+a^{\dagger}(f)}{\sqrt{2}}}
\end{equation*}
and define the {\it Weyl operators}
$W(f):=\exp{i\Phi(f)}$. The Weyl operators are unitary and satisfy the rules
\begin{align}
\label{cccrr1}
&W(h)^*=W(-h)\,,\quad W(0)=\idd\,,\\
&W(f)W(g)=e^{i\frac{\im(f,g)}{2}}W(f+g)\,,\quad f,g\in\gph\nn\,.
\end{align}
The abstract $C^*$--algebra  $\gw(\gph)$ generated by $\{W(f)\mid f\in\gph\}$ together with relations \eqref{cccrr1}, is simple and is commonly referred as the $C^*$--algebra of the CCR in the Weyl form. We always refer to $\gw(\gph)$ as the CCR algebra on $\gph$.

Let $H$ be a self--adjoint operator acting on $\bar\gph$. Suppose that $e^{itH}\gph\subset\gph$.  
The one--parameter unitary
group $T_tf:=e^{itH}f$ defines a one--parameter group of $*$--automorphisms $\a_t$ of $\gw(\gph)$ by putting $\a_t(W(f)):=W(T_t f)$. The latter is called the one--parameter group of {\it Bogoliubov automorphisms} generated by $T_t$.

A representation $\pi$ of the Weyl algebra $\gw(\gph)$ is {\it regular} if
the unitary group $\l\in\br\mapsto\pi(W(\l f))$ is continuous in the
strong operator topology, for any $f\in\gph$.  A state $\f$ on $\gw(\gph)$
is regular if the associated GNS representation is regular. 

The {\it quasi--free} states on the Weyl algebra are those of interest for our purposes. Such states 
$\om$ are uniquely determined by the two--point functions $\om(a^{\dagger}(f)a(g))$, $f,g\in\gph$.
The expectation value of a quasi--free state on the Weyl unitaries is easily recovered as
$$
\om(W(f))=\exp\{-[\|f\|^2/4+\om(a^{\dagger}(f)a(f))/2]\}\,.
$$
A standard textbook for CCR is \cite{BR2} (cf. Section 5.2) to which the reader is referred for proofs, literature and further details.

Let $G$ be any graph. We denote by $G$ itself the set of vertices $VG$ when this causes no confusion.
Suppose that $\gph$ contains the indicator functions $\{\d_x\mid x\in G\}\subset\ell^2(G)$.  A representation $\pi$ of $\gw(\gph)$ is said to be {\it
locally normal (w.r.t. the Fock representation)} if
$\pi\lceil_{\gw(\ell^2(\La))}$ is quasi--equivalent to the Fock
representation of $\gw(\ell^2(\La))$, $\La\subset G$ being any finite region.  A state on $\gw(\gph)$ is said to be locally
normal if the associated GNS representation is locally normal.  A locally normal
state $\f$ does have finite local density
\begin{equation*} 
\r_{\La}(\f):=\frac{1}{|\La|}
\sum_{j\in\La}\f(a^{\dagger}(\d_{j})a(\d_{j}))
\end{equation*}
even if the mean density  might be infinite (e.g. 
$\limsup_{\La\uparrow G}\r_\La(\f)=+\infty$). 
Let $\La_{n}\uparrow G$ be an {\it exhaustion} of $G$, that is
a sequence of finite regions invading the graph $G$, together with
a sequence of quasi--free states $\{\om_{\La_{n}}\}$ on $\gw(\ell^2(\La_n))$. It is shown in Lemma 3.2 of \cite{FGI1} that if  
\begin{equation*} 
\lim_{n}\om_{\La_{n}}(a^{\dagger}(\d_{j})a(\d_{j}))=+\infty
\end{equation*}
for some $j\in G$, then $\om(W(v)):=\lim_n\om_{\La_{n}}(W(v))$ cannot define any locally normal state on $\gw(\gph)$.

Let $G$ be any graph, equipped with the finite volume exhaustion $\{\La_n\}_{n\in\bn}$ kept fixed during the analysis. 
Fix a general bounded positive Hamiltonian $H\in\cb(\ell^2(G))$ admitting the IDS $N$ w.r.t. the fixed exhaustion $\{\La_n\}_{n\in\bn}$. As before, we normalise $H$ such that $0\in\s(H)$ is the bottom of the spectrum $\s(H)$ of $H$. Consider also the finite volume IDS $N_n$ relative to $H_n:=P_nHP_n$.
Define
\begin{align*}
\eps_{0}(H):=&\lim_{\La_{n}\uparrow G}
 \bigg(\inf\supp\big(N_n)\bigg)=0\,,\\
E_{0}(H):=&\inf\supp\bigg(\lim_{\La_{n}\uparrow G}
N_n\bigg) \equiv\inf\supp\big(N\big)\,.
\end{align*}
Here, the first limit exists by Lemma 3.4 of \cite{FGI1} and is 0 by normalisation. The second one is meaningful directly by the definition of the IDS.
We always get $E_0(H)\geq \eps_0(H)$. We then report the definition of the appearance of Hidden Spectrum, firstly introduced in 
\cite{BCRSV}.
\begin{Dfn} 
\label{hyspe}
For any one--particle Hamiltonian as above, if $\eps_{0}(H)< E_{0}(H)$ we say
that $H$ has the (low energy) {\it Hidden Spectrum}.
\end{Dfn}
The (mean) density of particles at the inverse temperature $\b>0$ and chemical potential $\m<\eps_{0}(K)$ for the Hamiltonian $K$ (where $K$ is either $H$ or its finite volume approximations $P_nHP_n$), is defined as
\begin{equation*} 
\r_{K}(\b,\m):=\int\frac{dN_{K}(h)}{e^{\b(h-\m)}-1}\,.
\end{equation*}
The corresponding {\it critical density} for the infinite volume Hamiltonian $H$ is given by
\begin{equation*} 
\r^{H}_{c}(\b):=\int\frac{dN_{H}(h)}{e^{\b h}-1}\equiv \r_{H}(\b,0)\,,
\end{equation*}
see e.g. \cite{BR2, LL}.\footnote{It is also customary to fix the {\it  activity} $z:=e^{\b\m}$, instead of the chemical potential.}
The following theorem collects the implications of the appearance of Hidden Spectrum, relatively to the critical density.   
\begin{Thm}
For an Hamiltonian $H\geq0$ on a graph $G$ with $0\in\s(H)$ admitting the IDS $N$, Hidden Spectrum implies that the critical density is always finite. 

If $H$ is the Pure Hopping one--particle Hamiltonian of a negligible additive perturbation graph $Y$ of another graph $X$ admitting IDS, then
$$
\r^{H_Y}_{c}(\b)=\r_{H_X}(\b,\d)\leq\r^{H_X}_{c}(\b)\,.
$$
with $\d=\|A_X\|-\|A_Y\|$.
\end{Thm}
\begin{proof}
Suppose that there is Hidden Spectrum. This means that for some $\eps>0$
$$
\r^{H}_{c}(\b)=\int_\eps^{+\infty}\frac{dN_{H}(h)}{e^{\b h}-1}<+\infty\,.
$$
The second part follows directly by Theorem \ref{density0}.
\end{proof}
The second part of the previous theorem explain in the case of the appearance of Hidden Spectrum for the Pure Hopping model, that $\d$ given in  \eqref{ddee} plays the role of a chemical potential.
 \begin{figure}
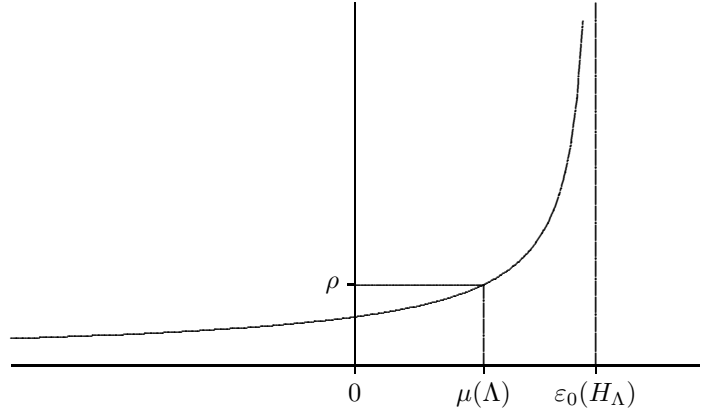

\hbox to\hsize\bgroup\hss
\beginpicture
\setcoordinatesystem units <1.8truein,.2truein>
\setplotarea x from -1 to 1, y from 0 to 9.5
\axis bottom ticks short withvalues {$0$} {$\m(\La)$} {$\eps_0(H_\La)$}
/ at 0 0.3745455  0.7 / /
% {$\scriptstyle \Lambda(x)$}
% {$\scriptstyle 1$} / at 0  .4  .7  1 / /
\axis left shiftedto x=0 ticks short withvalues {$\rho$} / at 2.1135204 / /
\plot  0 2.1135204  0.3745455 2.1135204 /
\plot 0.7 0  0.7 9.5 /
\plot  0.3745455 0  0.3745455 2.1135204 /
\plot
-1.           0.7020489
  -0.9830303    0.7067601
  -0.9660606    0.7115512
  -0.9490909    0.7164242
  -0.9321212    0.7213816
  -0.9151515    0.7264256
  -0.8981818    0.7315588
  -0.8812121    0.7367835
  -0.8642424    0.7421026
  -0.8472727    0.7475187
  -0.8303030    0.7530348
  -0.8133333    0.7586537
  -0.7963636    0.7643787
  -0.7793939    0.7702129
  -0.7624242    0.7761597
  -0.7454545    0.7822226
  -0.7284848    0.7884054
  -0.7115152    0.7947118
  -0.6945455    0.8011458
  -0.6775758    0.8077116
  -0.6606061    0.8144137
  -0.6436364    0.8212565
  -0.6266667    0.8282449
  -0.6096970    0.8353839
  -0.5927273    0.8426788
  -0.5757576    0.850135
  -0.5587879    0.8577584
  -0.5418182    0.8655550
  -0.5248485    0.8735312
  -0.5078788    0.8816938
  -0.4909091    0.8900497
  -0.4739394    0.8986065
  -0.4569697    0.9073719
  -0.44         0.9163543
  -0.4230303    0.9255624
  -0.4060606    0.9350053
  -0.3890909    0.9446927
  -0.3721212    0.9546351
  -0.3551515    0.9648433
  -0.3381818    0.9753288
  -0.3212121    0.9861039
  -0.3042424    0.9971817
  -0.2872727    1.0085759
  -0.2703030    1.0203012
  -0.2533333    1.0323734
  -0.2363636    1.0448092
  -0.2193939    1.0576263
  -0.2024242    1.0708439
  -0.1854545    1.0844823
  -0.1684848    1.0985635
  -0.1515152    1.1131108
  -0.1345455    1.1281495
  -0.1175758    1.1437067
  -0.1006061    1.1598115
  -0.0836364    1.1764955
  -0.0666667    1.1937928
  -0.0496970    1.2117401
  -0.0327273    1.2303776
  -0.0157576    1.2497487
    0.0012121    1.2699007
    0.0181818    1.2908852
    0.0351515    1.3127586
    0.0521212    1.3355828
    0.0690909    1.3594257
    0.0860606    1.384362
    0.1030303    1.4104742
    0.12         1.4378536
    0.1369697    1.4666016
    0.1539394    1.496831
    0.1709091    1.5286679
    0.1878788    1.5622534
    0.2048485    1.5977467
    0.2218182    1.6353271
    0.2387879    1.6751982
    0.2557576    1.7175919
    0.2727273    1.7627737
    0.2896970    1.8110489
    0.3066667    1.8627711
    0.3236364    1.9183517
    0.3406061    1.9782727
    0.3575758    2.0431033
    0.3745455    2.1135204
    0.3915152    2.1903362
    0.4084848    2.2745348
    0.4254545    2.3673211
    0.4424242    2.470188
    0.4593939    2.5850104
    0.4763636    2.7141784
    0.4933333    2.8607926
    0.5103030    3.0289567
    0.5272727    3.2242303
    0.5442424    3.454357
    0.5612121    3.7304879
    0.5781818    4.069352
    0.5951515    4.4973694
    0.6121212    5.059147
    0.6290909    5.8371213
    0.6460606    7.0048035
    0.6630303    9.0109256
%    0.68         13.572088
    /
\endpicture
\hss\egroup
 \caption{{The finite volume chemical potential at fixed density $\r$.}}
     \label{figg1}
\end{figure}

It is well known that the condensation regime is described by $\m=0$, see e.g. Section 5.2 of \cite{BR2}, or Section 3 of \cite{FGI1}. 
The most common way to study the appearance of BEC in homogeneous systems (cf. Fig. \ref{figg1}), is to determine the chemical potential
$\m(\La)$ at finite volume $\La\subset G$, and at fixed inverse temperature $\b$, as the unique solution of the equation
\begin{equation}
\label{fvden}
 \r_{H_{\La}}(\b,\m(\La))=\r\,,
 \end{equation}
where $\r$ is the mean density of the system fixed a--priori. 
The infinite volume limit can be investigated by means of the reference exhaustion 
$\{\La_n\}_{n\in\bn}$, with $H_{\La_n}=P_nHP_n$. To take into account also the very different situation appearing in non homogeneous situation, we can also start by fixing any general sequence of chemical potential 
$$
\m_n<\eps_{0}(H_{\La_n})\,,
$$
which we can suppose to converge (eventually passing to a subsequence) to some $\m$. In the case when such a sequence is recovered by using \eqref{fvden}, we put $\m_n:=\m(\La_n)$. Since $\eps_{0}(H_{\La_n})\to\eps_{0}(H)=0$, we get $\m\leq0$.
 The finite volume state with density $\rho$ at $\b$ is described by
 the two--point function
 \begin{equation} 
 \label{sato}
     \om_{\La_{n}}(a^{\dagger}(\xi)a(\eta))=\big\langle  
     (e^{\b(H_{\La_{n}}-\m_{n}\idd)}-\idd)^{-1}\xi,\eta\big\rangle,
 \end{equation}
 where $\xi,\eta\in\ell^{2}(\La_n)$.
After extending the states with two--point function in \eqref{sato} to the whole network, infinite volume states are 
described as cluster points of the net of such finite volume states. 
Concerning the infinite volume limit of the finite volume density, we get (cf. \cite{BCRSV, F2, FGI1}) in the condensation regime for each sequence of the chemical potential $\m_n\to0$,
$$
\lim_{\eps\downarrow0}\lim_{\La_{n}\uparrow G}
\int\frac{F_\eps(x)}{e^{\b(x-\m_{n})}-1}\di N_{H_{\La_{n}}}(x)\,,
=\r_c^H(\b)\,.
$$
provided $\lim_{x\downarrow0}N_H(x)=0$. Here, $\{F_\eps\}_{\eps>0}$ is made of continuous mollifiers, all vanishing in a neighborhood of $0$ and converging almost everywhere (w.r.t the measure determined by $N_H$) to $1$, see Proposition 3.3 of \cite{F2}. 
Thus, the quantity
 $$
 n_{0}:=\lim_{\eps\downarrow0}\lim_{\La_{n}\uparrow X}
 \int\frac{1-F_{\eps}(h)}{e^{\b(h-\m_{n})}-1}dN_{H_{\La_{n}}}(h)
 $$
 is well defined and independent of the particular choice of the
 mollifiers as above, and describes the amount of the condensate in the ground state.  Indeed, if the sequence of chemical potentials are obtained through \eqref{fvden} by fixing the mean density $\r\geq\r_c(\b)$, we have
 \begin{equation*}
\r=n_{0}+\r_c(\b)\,.
 \end{equation*}
Among the new results proved in the forthcoming sections for the inhomogeneous models under consideration, we prove that it can happen that the amount of the condensate converges in the infinite volume limit, even if for the sequence of the finite volume two--point function in \eqref{sato},
$$
\lim_n\om_{\La_{n}}(a^{\dagger}(\d_x)a(\d_x))=+\infty\,.
$$ 
The reader is also referred to \cite{BLP} and the reference therein, for the investigation of some properties connected with BEC, when other exhaustions different from the standard ones are considered. In such an analysis, no limit of the two--point function 
$\om_{\La_{n}}(a^{\dagger}(f)a(f))$ has been investigated.

To simplify the foregoing analysis without affecting the obtained results, from now on we put $\b=1$ if it is not otherwise specified. We also put $\r(\m):=\r^H(1,\m)$,  
$\r_{\La}(\m):=\r^{H_{\La}}(1,\m)$, and finally $\r_c:=\r^H(1,0)$, $H$ being the Hamiltonian of the model, which will always be the Pure Hopping one.

\section{General results}
\label{sec:gen}

We start with a simple result useful in the sequel without further mention, which allows us to manage the resolvent $R_{A_X}$ of the unperturbed Adjacency of the finite volume graph $X$ entering in the Krein Formula \eqref{k2s} of the perturbed Adjacency $R_{A_Y}$, for values of $\l>\|A_Y\|$.
\begin{Lemma}
\label{caficz}
Let $Y$ be an additive perturbation of $X$ involving only edges, with $|VX|<+\infty$. Then $\|A_Y\|>\|A_X\|$.
\end{Lemma}
\begin{proof}
Let $v\in\ell^2(VX)$ be any normalised PF eigenvector for $A_X$. Put $A_Y=A_X+D$.
\begin{align*}
\|A_Y\|\geq\langle A_Y v,v\rangle=&\langle (A_X+D)v,v\rangle=\langle A_X v,v\rangle
+\langle Dv,v\rangle\\
=&\|A_X\|+\langle Dv,v\rangle>\|A_X\|\,.
\end{align*}
\end{proof}
Let $Y$ be any additive perturbation of $X$ involving only edges described by $D$. Define $Z\subset Y$ as the subgraph whose vertices are precisely the sources (or equivalently ends, being the networks under consideration unordered) of the additional edges. The perturbation $D=P_{\ell^2(VZ)}DP_{\ell^2(VZ)}$ can be directly viewed by acting on $\ell^2(VZ)$. The candidate for the {\it Secular Equation }(cf. Theorem 3.1 of \cite{F} and Theorem 6.1 of \cite{FGI1}) in this more general situation might assume the form
\begin{equation}
\label{eksek}
\spr\big(DR_{A_X}(\l)P_{\ell^2(VZ)}\big)=1\,.
\end{equation}
More in detail, 
\eqref{eksek} would have at most one solution 
$\l_*\in(\|A_X\|,+\infty)$ which necessarily gives $\|A_Y\|=\l_*>\|A_X\|$, otherwise $\|A_Y\|=\|A_X\|$. The proof would follow the same lines of Theorem 3.2 of \cite{F}, after showing that the function
$$
\l\in(\|A_X\|,+\infty)\mapsto\spr\big(DR_{A_X}(\l)P_{\ell^2(VZ)}\big)
$$
is nondecreasing and strictly convex, by using the $1^{st}$ identity of resolvent and (1.1) of \cite{S}. 
As we do not use this result in the sequel, we have chosen not to pursue more such a possible generalisation. However, by using the same lines of the mentioned theorem,  we can prove that if $\l_*>\|A_X\|$ is the solution (necessarily unique) of
\begin{equation}
\label{1eksek1}
\big\|DR_{A_X}(\l)P_{\ell^2(VZ)}\big\|=1\,,
\end{equation}
then $\l_*=\|A_Y\|>\|A_X\|$. Conversely, if \eqref{1eksek1} has no solution $\l>\|A_X\|$, then $\|A_Y\|=\|A_X\|$. In other words, it can be proven that \eqref{eksek} is indeed the Secular Equation for the case $DR_{A_X}(\l)P_{\ell^2(VZ)}$ self--adjoint.

We pass to the definition of the geometrical and PF dimensions, the latter, appeared first in \cite{FGI1}, takes into account the 
growth of the $\ell^2$--norm of a fixed PF weight of the Adjacency, in amenable cases. It plays a fundamental role in the study of the condensation effects, see below.
Let $X$ be a graph, together with an exhaustion $\{\La_n\}_{n\in\bn}$ and a PF weight  $\{v(x)\mid x\in X\}$ for the adjacency.
\begin{Dfn}${}$
\label{pfcazol}
\begin{itemize}
\item[(i)] The {\it geometrical dimension} $d_G(X)$ of $X$, is defined to be $g\geq0$ if and only if $|\La_n|\approx a^g$ for $n\to\infty$, for some number $a>0$. 
\item[(ii)] The {\it Perron--Frobenius dimension} $d_{PF}(X)$ of $X$, is defined as $p\geq0$ if and only if 
$\|v\lceil_{\ell^2(V\La_n)}\|^2\approx b^p$ for $n\to\infty$, for some number $b>0$.
\end{itemize}
\end{Dfn}
In the sequel, we omit the possible dependence on the chosen exhaustion and on the PF weight when it causes no confusion.

We end with a result of general nature concerning the existence/non existence of locally normal KMS states exhibiting BEC for the Pure Hopping model on general networks. For the sake of completeness, we first report Proposition 5.2 of \cite{F2} concerning the recurrent situation. 
Let $\{\om_{\La_n}\}_{n\in\bn}$ be the sequence of finite volume quasi--free states whose two--point function is given in \eqref{sato} for any sequence of chemical potentials $\{\m_n\}_{n\in\bn}$, such that $\m<\|A_G\|-\|A_{\La_n}\|$ and $\lim_n\m_n=0$.
\begin{Prop}
\label{roec}
If $A_G$ is recurrent, then for each $x\in G$,
$$
\lim_n\om_n(a^\dagger(\d_x)a(\d_x))=+\infty\,.
$$
\end{Prop}
Define 
\begin{equation*}
\gph:=\ssv\left\{e^{\imath t(\|A_G\|\idd-A_G)}\d_x\mid t\in\br\,,x\in G\right\}
\end{equation*}
for the linear span (no closure in $\ell^2(G))$ of elements of the form $e^{\imath t(\|A_G\|\idd-A_G)}\d_x$. Fix any 
PF weight $v$ for $A_G$ which exists by compactness (cf. Section 4 of \cite{FGI1}). Consider the continuous function
\begin{equation}
\label{resam}
f(x)= 
     \begin{cases}
     -\frac12\,,&x=0\,,\\
\frac{1}{e^x-1}-\frac1x\,,&x>0\,.
     \end{cases}
\end{equation}
It is bounded on $[0,+\infty)$, and 
$$
\frac{1}{e^x-1}=f(x)+\frac1x\,,\quad x\in(0,+\infty)\,.
$$
It provides the precise comparison heuristically explained in \eqref{euro}, between the Resolvent of $A$ and the functional calculus of $(e^H-\idd)^{-1}$ associated to the Bose--Gibbs occupation number \eqref{bgokn} for the Pure Hopping Hamiltonian.
\begin{Thm}
\label{tretre}
For a graph $G$, suppose that $A_G$ is transient. Fix any PF weight $v$ for $A_G$. For each $D\geq0$, the two--point function
\begin{align}
\label{kmsbecd}
\om_D(a^\dagger(u_1)a(u_2)):=&\big\langle(e^{\b(\|A_G\|\idd-A_G)}-\idd)^{-1}u_1,u_2\big\rangle\nn\\
+&D\langle u_1,v\rangle\langle v, u_2\rangle\,,
\quad u_1,u_2\in\gph\,.
\end{align}
uniquely defines locally normal KMS states on the Weyl CCR algebra $\gw(\gph)$ w.r.t. the dynamics generated by the Bogoliubov transformations $u\in \gph\mapsto e^{\imath t(\|A_G\|-A_G)}u$, $t\in\br$, $u\in\gph$.
\end{Thm}
\begin{proof}
By construction, $e^{\imath tH}\gph=\gph$, Thus, $t\mapsto e^{\imath tH}$ defines a one--parameter group of Bogoliubov automorphisms of 
$\gw(\gph)$, for which the states determined by the two--point function \eqref{kmsbecd} are KMS, provided that the r.h.s. is 
well--defined. By considering the functional calculus of the function \eqref{resam}, we can reduce the matter to the Resolvent $R_{A_G}(\l)$. For the generator $u=e^{\imath tH}\d_x$ and $\l\downarrow\|A_G\|$ , we get
$$
\left\langle R_{A_G}(\l)e^{\imath t(\|A_G\|-A_G)}\d_x,e^{\imath t(\|A_G\|-A_G)}\d_x\right\rangle
=\left\langle R_{A_G}(\l)\d_x,\d_x\right\rangle\uparrow\left\langle R_{A_G}(\|A_G\|)\d_x,\d_x\right\rangle\,,
$$
which is finite because $A_G$ as transient. Thus, $e^{\imath tH}\d_x\in\cd\left((e^{\b H}-\idd)^{-1/2}\right)$ and the first addendum in the r.h.s. of \eqref{kmsbecd} is meaningful. Concerning the second one describing the portion of the condensate, for the circle $C_r\subset\bc$ of sufficiently big radius $r$, we note that if 
$z=re^{\imath\th}\in C_r$, then
$$
|\langle R_{A_G}(z)\d_x,\d_y\rangle|=\left|\sum_{n=0}^{+\infty}\frac{\langle A^n_G\d_x,\d_y\rangle}{r^{n+1}}e^{-\imath(n+1)\th}
\right|
\leq\sum_{n=0}^{+\infty}\frac{\langle A^n_G\d_x,\d_y\rangle}{r^{n+1}}=\langle R_{A_G}(r)\d_x,\d_y\rangle\,.
$$
For $u\in\ell^2(G)$, denote $|u|\in\ell^2(G)$ the vector whose entries are defined as $|u|(x):=|u(x)|$, $x\in VG$.
By reasoning as in Proposition 4.4 of \cite{F2}, we get
\begin{align*}
\big\langle\big|e^{\imath tH}\d_x&\big|,v\big\rangle=
\left\langle\left|\frac{1}{2\pi\imath}\oint_{C_r}e^{\imath t(\|A_G\|-z)}R_{A_G}(z)\d_x\, \di z
\right|,v\right\rangle\\
=&\sum_{y\in G}\left|\frac{1}{2\pi\imath}\oint_{C_r}e^{\imath t(\|A_G\|-z)}\langle R_{A_G}(z)\d_x,\d_y\rangle \di z\right|v(y)\\
\leq&r\sum_{y\in G}\langle R_{A_G}(r)\d_x,\d_y\rangle v(y)
\leq r\left\langle R_{A_G}(r)\d_x,v\right\rangle\\
=&r\left\langle\d_x,R_{A_G}(r)v\right\rangle
=\frac{r\left\langle\d_x,v\right\rangle}{r-\|A_G\|}
=\frac{rv(x)}{r-\|A_G\|}\,.
\end{align*}
Namely, if $u\in\gph$ then $\sum_{y\in G}|u(y)|v(y)<+\infty$ as $u=\sum_{x\in G}a_xe^{\imath tH}\d_x$ is a finite sum, that is the last addendum in the r.h.s. of \eqref{kmsbecd} is also meaningful.
\end{proof}
The last results explain the remarkable unexpected fact that the BEC for the Pure Hopping model is connected with the transience/recurrence character of the Adjacency, and not with the finiteness of the critical density. It might be straightforwardly generalised to models on graphs, as well as on $\br^d$ for semibounded Schr\"dinger Hamiltonians of the form $H=-\D+V(x)$. Below, we exhibit models whose critical density is infinite but exhibiting BEC and vice--versa, that is those for which the critical density is finite but there is no locally normal state exhibiting BEC. 

We end the present section by noticing that, for any finite subgraph $\La\subset G$ and any chemical potential 
$\m<\|A_G\|-\|A_\La\|$, the finite volume amount of the condensate for finite volume states in \eqref{sato} and states \eqref{kmsbecd}, defined respectively as 
$$
\r^{cond}(\om_{\La}):=\frac1{|\La|(\|A_G\|-\|A_\La\|-\m)}\,,\quad
\r^{cond}_{\La}(\om):=\frac{D\sum_{x\in\La}v(x)^2}{|\La|}\,,
$$
are meaningful.

\section{Comb graphs}
\label{sec:comb}

The present section is devoted to general results relative to the harmonic analysis for the Adjacency on the so--called {\it Comb Graphs}, the last being the main objects of the investigation of the remaining part of the present paper. 
\begin{Dfn}
     Let $G$, $H$ be graphs, and let $o\in VH$ be a given vertex. 
     Then the {\it comb product} $Y := G \comb (H,o)$ is a graph with $VY:=
     VG \times VH$, and $(g,h)$, $(g',h') \in VY$ are adjacent $\iff$
     $g=g'$ and $h\sim h'$ or $h=h'=o$ and $g\sim g'$.  We call $G$ the {\it base graph}, and
     $H$ the {\it fibre graph}.  
    \end{Dfn} 

When $o\in H$ is understood from the
context, we omit it and write $G\comb H$. 
Notice that $\ell^2(VG\comb (H,o))=\ell^2(VG)\otimes\ell^2(VH)$. In this case, $G\comb H$ can be viewed as the additive perturbation of the disjoint union 
$X:=\bigsqcup_{VG}H$
of $\# VG$--copies of $H$. For $D$ describing the perturbation, we have $D=A_G\otimes P_o$, 
$P_o:=\langle\,{\bf\cdot}\,,\d_o\rangle\d_o$ being the orthogonal projection onto the subspace generated by 
$\d_o\in\ell^2(VH)$. 
\begin{Prop}
\label{infneg}
If $|VH|=+\infty$ then $G\comb (H,o)$ is a negligible additive perturbation of $\bigsqcup_{VG} H$. 
\end{Prop}
\begin{proof}
Let $\{G_n\}_{n\in\bn}$, $\{H_n,o\}_{n\in\bn}$ exhaustions of $G$, $(H,o)$ respectively, with $o\in H_n$, $n\in\bn$.
Then
$$
\frac{\big|E\big(G_n\comb (H_n,o)\big)\backslash E\big(\bigsqcup_{VG_n} H_n\big)\big|}{\big|V\big(\bigsqcup_{VG_n} H_n\big)\big|}
=\frac{|EG_n|}{|VG_n||VH_n|}\leq\frac{\deg_G|VG_n|}{|VG_n||VH_n|}=\frac{\deg_G}{|VH_n|}\to0\,.
$$
\end{proof}
The following result useful in the sequel, concerns the explicit expression of the Krein Formula for the Adjacency of Comb Graphs.
\begin{Prop}
\label{risgfor}
Let $\l\in\{z\mid |z|>\|A_{G \comb (H,o)}\|\}$. 
Then
\begin{equation}
 \label{nn0110}
R_{A_{G \comb (H,o)}}(\l)=\idd_{\ell^2(VG)}\otimes R_{A_H}(\l)+g(\l)R_{A_G}(g(\l))A_G\otimes R_{A_H}(\l)P_oR_{A_H}(\l)\,,
\end{equation}
where 
\begin{equation}
\label{eksek3}
g(\l):=\langle R_{A_H}(\l)\d_o,\d_o\rangle^{-1}\,.
\end{equation}
\end{Prop}
\begin{proof}
The proof follows by a direct application of the definition of the comb product to the Krein Formula \eqref{k2s} to the complement of the disk of radius $\|A_{G \comb (H,o)}\|$ on which $I_{\ell^2(G)}-S(\l)$ is certainly invertible, by taking into account that 
$\ell^2(G \comb (H,o))=\ell^2(G)\otimes\ell^2(H)$, see e.g. the proof of Proposition 9.5 of \cite{FGI1}.
\end{proof}
Using the Secular Equation \eqref{1eksek1}, another relevant step is to decide whether the norm of the Adjacency of the Comb Graph $G \comb (H,o)$ is greater than that of 
$H$. 
\begin{Prop}
\label{eksek2}
The equation
\begin{equation}
\label{eksek1}
\langle R_{A_H}(\l)\d_o,\d_o\rangle\|A_G\|=1
\end{equation}
has at most one solution $\l_*>\|A_H\|$. If there is no solution $\l_*>\|A_H\|$, then $\|A_{G \comb (H,o)}\|=\|A_H\|$. If such a solution $\l_*>\|A_H\|$
exists, then
$\l_*=\|A_{G \comb (H,o)}\|$. 
\end{Prop}
\begin{proof}
In this situation $Z$ coincides with $G$, and $X$ is the disjoint union of $\# VG$--copies of $H$. Then 
$$
DR_{A_X}(\l)P_{\ell^2(VZ)}=\langle R_{A_H}(\l)\d_o,\d_o\rangle A_G
$$ 
acting directly on $\ell^2(VZ)$, which is symmetric.
The proof now follows the lines of the analogous Theorem 3.2 of \cite{F}. 
\end{proof}
The foregoing step concerns the construction of a PF weight $v$ for the adjacency of the comb graph $G \comb (H,o)$, starting from a PF one $w$ for $A_G$.
\begin{Prop}
\label{prcfi}
For the Comb Product $G \comb (H,o)$ and $w\in\cw_G$, the following assertions hold true.
\begin{itemize}
\item[(i)] Suppose that $\langle R_{A_H}(\|A_H\|)\d_o,\d_o\rangle\geq\|A_G\|^{-1}$. Then $A_{G \comb (H,o)}$ is transient if and only if $A_G$ is transient. In addition, 
\begin{equation*}
v=w\otimes R_{A_H}(\|A_{G \comb (H,o)}\|)\d_o
\end{equation*}
provides a PF weight for $A_{G \comb (H,o)}$, which is the unique PF weight up to a multiplicative constant if $A_{G \comb (H,o)}$ is recurrent.
\item[(ii)] Suppose that $\langle R_{A_H}(\|A_H\|)\d_o,\d_o\rangle<\|A_G\|^{-1}$. Then $A_{G \comb (H,o)}$ is always transient. \end{itemize}
\end{Prop}
\begin{proof}
$(i)$  If $\langle R_{A_H}(\|A_H\|)\d_o,\d_o\rangle\geq\|A_G\|^{-1}$, then the smooth function in \eqref{eksek3}
$g(\l)\downarrow\|A_H\|$ if $\l\downarrow\|A_{G \comb (H,o)}\|$. By 
Proposition \ref{eksek2}, either $\langle R_{A_H}(\|A_H\|)\d_o,\d_o\rangle>\|A_G\|^{-1}$, which means that 
$\|A_{G \comb (H,o)}\|>\|A_H\|$, or $\langle R_{A_H}(\|A_H\|)\d_o,\d_o\rangle=\|A_G\|^{-1}$ which means that $A_H$ is transient.
In both situations, $\langle R_{A_H}(\|A_{G \comb (H,o)}\|)\d_o,\d_o\rangle$ is always finite. Thus,
$A_{G\comb (H,o)}$ is transient if and only if 
$A_G$ is transient by Krein Formula for the resolvent \eqref{nn0110}. Concerning the PF weight,
first we note that, if $A_H$ is transient, then $\langle R_{A_H}(\|A_H\|)\d_h,\d_o\rangle<+\infty$. With $|E_{h,o}|$ the number of edges connecting $h$ and $o$, and
$$
\frac{x}{\l-x}=\frac{\l}{\l-x}-1\,,
$$
it automatically follows from the fact that 
\begin{align*}
\sum_{h\sim o}|E_{h,o}|\langle R_{A_H}(\|A_H\|)&\d_h,\d_o\rangle=\langle A_HR_{A_H}(\|A_H\|)\d_o,\d_o\rangle\\
=&\|A_H\|\langle R_{A_H}(\|A_H\|)\d_o,\d_o\rangle-1\,,
\end{align*}
where the l.h.s. is made of a finite numbers of addenda, and the r.h.s. is bounded. To simplify the matter, we suppose that 
$\langle R_{A_H}(\|A_H\|)\d_o,\d_o\rangle>\|A_G\|^{-1}$. The case $\langle R_{A_H}(\|A_H\|)\d_o,\d_o\rangle=\|A_G\|^{-1}$ is straightforwardly obtained at the same way, by considering an exhaustion $\{(H_n, o)\}_{n\in\bn}$ also for $(H_n, o)$ and taking into account Lemma \ref{caficz} together with the fact that $A_H$ is transient.
Choose an exhaustion $\{G_n\}_{n\in\bn}$ for $G$, together with the sequence
$\{w_n\}_{n\in\bn}$ of the PF eigenvectors, all normalised as $w_n(\g)=w(\g)$ at a common root $\g\in G_n$, and extended to 0 on $VG\backslash VG_n$ such that $\lim_n w_n(g)=w(g)$, $g\in G$. It certainly exists as $w\in\cw_G$. As $\langle R_{A_H}(\|A_H\|)\d_o,\d_o\rangle>\|A_G\|^{-1}$, we can also suppose that 
$\|A_{G_n\comb (H,o)}\|>\|A_H\|$ for each $n\in\bn$. Theorem 6.1 of \cite{FGI1} assures that 
$$
v_n(g,h)=\langle R_{A_H}(\|A_{G_n \comb (H,o)}\|)\d_h,\d_o\rangle w_n(g)\,, \quad (g,h)\in VG_n\times VH\,,
$$
provides a PF eigenvector (unique with the chosen normalisation) for $G_n \comb (H,o)$.
As 
$\langle R_{A_H}(\|A_H\|)\d_o,\d_o\rangle>\|A_G\|^{-1}$, then 
$\|A_{G \comb (H,o)}\|>\|A_H\|$ and $R_{A_H}(\|A_{G_n \comb (H,o)}\|)\rightarrow R_{A_H}(\|A_{G \comb (H,o)}\|)$ in norm, as 
$\|A_{G_n \comb (H,o)}\|\uparrow\|A_{G \comb (H,o)}\|$. 
Thus, $v_n(g,h)\rightarrow v(g,h)$, point--wise. By taking into account that
$\|A_{G_n\comb (H,o)}\|\uparrow\|A_{G\comb (H,o)}\|\geq\|A_H\|$, we compute with
$$
A_{G\comb (H,o)}=A_G\otimes P_o+I\otimes A_H\,,
$$
and the analogous for the $A_{G_n\comb (H,o)}$,
\begin{align*}
\langle A_{G\comb (H,o)}v,&\d_g\otimes\d_h\rangle
=\langle A_G w,\d_g\rangle\langle R_{A_H}(\|A_{G \comb (H,o)}\|)\d_o,\d_o\rangle\d_{h,o}\\
+&\langle w,\d_g\rangle\langle A_HR_{A_H}(\|A_{G \comb (H,o)}\|)\d_o,\d_h\rangle\\
=&\|A_G\|\langle w,\d_g\rangle\langle R_{A_H}(\|A_{G \comb (H,o)}\|)\d_o,\d_o\rangle\d_{h,o}\\
+&\langle w,\d_g\rangle\langle A_HR_{A_H}(\|A_{G \comb (H,o)}\|)\d_o,\d_h\rangle\\
=&\lim_n\bigg(\|A_{G_n}\|\langle w_n,\d_g\rangle\langle R_{A_H}(\|A_{G_n \comb (H,o)}\|)\d_o,\d_o\rangle\d_{h,o}\\
+&\langle w_n,\d_g\rangle\langle A_HR_{A_H}(\|A_{G_n \comb (H,o)}\|)\d_o,\d_h\rangle\bigg)\\
=&\lim_n\bigg(\langle A_{G_n}w_n,\d_g\rangle\langle R_{A_H}(\|A_{G_n \comb (H,o)}\|)\d_o,\d_o\rangle\d_{h,o}\\
+&\langle w_n,\d_g\rangle\langle A_HR_{A_H}(\|A_{G_n \comb (H,o)}\|)\d_o,\d_h\rangle\bigg)\\
=&\lim_n\bigg(\langle A_{G_n\comb (H,o)}v_n,\d_g\otimes\d_h\rangle\bigg)\\
=&\lim_n\bigg(\|A_{G_n\comb (H,o)}\|\langle v_n,\d_g\otimes\d_h\rangle\bigg)\\
=&\|A_{G\comb (H,o)}\|\langle v,\d_g\otimes\d_h\rangle\,.
\end{align*}
Namely, the point--wise limit $\lim_n v_n(g,h)=v(g,h)$ provides a PF weight for $A_{G \comb (H,o)}$, which is unique if $A_{G\comb (H,o)}$ is recurrent (cf. \cite{S}), which happens if and only if $A_{G}$ is so.

$(ii)$ If $\langle R_{A_H}(\|A_H\|)\d_o,\d_o\rangle<\|A_G\|^{-1}$, then $A_H$ is transient, $\|A_{G \comb (H,o)}\|=\|A_H\|$, and
finally $\l\in[\|A_{G \comb (H,o)}\|,+\infty)\Rightarrow g(\l)>\|A_G\|$. Thus, 
$\langle R_{A_G}(g(\|A_{G \comb (H,o)}\|))A_G\d_{g},\d_{g}\rangle$ is finite for each $g\in G$, and
$\langle R_{A_H}(\|A_{G \comb (H,o)}\|)\d_o,\d_o\rangle$ is also finite because $A_H$ is transient. We then conclude that 
$A_{G \comb (H,o)}$ is always transient again by \eqref{nn0110}. 
\end{proof}
The final part of the present section is devoted to the particular cases useful in the sequel, that is when $G$ and/or $H$ are isomorphic to 
$\bz^d$. For this purpose, we use the boundary conditions for the Adjacency $A_{\bz^d}$. Indeed, fix the segment 
$$
\S_n:=\{-n,-n+1,\dots,0,\dots,n-1,n\}\subset\bz
$$
of $\bz$ made of $2n+1$ points, together with the finite circle group $\bt_{2n+1}$ obtained by $\S_n$ by adding only an edge.
We always refer to such a graph with {\it periodic boundary condition} directly as $\bt_{2n+1}$.
By using Fourier transform, the circle group $\bt$ and its powers $\bt^d$,

are also considered as the dual of $\bz^d$. We identify $\bt\sim[-\pi,\pi]$, the latter equipped with the sum operation modulus $2\pi$, and the normalised Haar measure $\frac{\di\th}{2\pi}$. 
 With an abuse of notations, the normalised Haar measures on $\bt^d_{2n+1}$ and $\bt^d$ are 
symbolically denoted by $\di m_n(\th)$, $\di m(\th)$ by omitting the dependence on the dimension $d$. 

The following result similar to Proposition \ref{infneg}, assures that for exhaustions with periodic boundary conditions for the bases and/or fiber spaces $G$, $H$, we still have additive finite volume approximations of the Adjacency which become negligible in the limit of infinite volume. 
\begin{Prop}
\label{pradneg}
Let $H=\bz^d$, and $G$ be equipped with the exhaustion $\{G_n\}_{n\in\bn}$. Then
$$
\lim_n\frac{|E(G_n\comb\bt^d_{2n+1})|-\big|E\big(\bigsqcup_{VG_n}\S^d_{n}\big)\big|}
{|V(G_n\comb\S^d_{n})|}=0\,.
$$
If in addition $G=\bz^\n$, again
$$
\lim_n\frac{|E(\bt^\n_{2n+1}\comb\bt^d_{2n+1})|-\big|E\big(\bigsqcup_{V\S^\n_n}\S^d_{n}\big)\big|}
{|V(\S^\n_n\comb\S^d_{n})|}=0\,.
$$
\end{Prop}
\begin{proof}
We start by noticing that for the adjacency of the graph $\G_1\times\G_2$,
$$
A_{\G_1\times\G_2}=A_{\G_1}\otimes I_{\ell^2(V\G_2)}+I_{\ell^2(V\G_1)}\otimes A_{\G_2}\,.
$$
In addition, if $D$ is a self--adjoint matrix describing a number of edges 
$$
n_D=\frac12\sum_{i,j}D_{ij}\,, 
$$
$D\otimes I_{\ell^2(J)}$ describes a number of edges
\begin{equation}
\label{perpre}
n_{D\otimes I_{\ell^2(J)}}=|J|n_D\,.
\end{equation}
To pass from $\S_n$ to $\bt_{2n+1}$, it is enough to add only one edge, that is 
$A_{\bt_{2n+1}}=A_{S_n}+D$ with
$n_D=1$. Thus,
$$
n_{A_{\bt^d_{2n+1}}}=n_{A_{\S^d_n}}+d(2n+1)^{d-1}\,.
$$
As $G$ is supposed of uniformly bounded degree, then
$$
\sup_n\frac{|EG_n|}{|VG_n|}<+\infty\,.
$$
A simple calculation yields
$$
\frac{|E(G_n\comb\bt^d_{2n+1})|-\big|E\big(\bigsqcup_{VG_n}\S^d_{n}\big)\big|}
{|V(G_n\comb\S^d_{n})|}
=\frac{|VG_n|d(2n+1)^{d-1}+|EG_n|}{|VG_n|(2n+1)^d}\rightarrow0\,.
$$
If we adopt the periodic boundary conditions also on the base space $G=\bz^\n$, by using again \eqref{perpre} and reasoning as above, we simply get 
$$
\frac{|E(\bt^\n_{2n+1}\comb\bt^d_{2n+1})|-\big|E\big(\bigsqcup_{V\S^\n_n}\S^d_{n}\big)\big|}
{|V(\S^\n_n\comb\S^d_{n})|}
=d\frac{(2n+1)^{d-1}+1}{(2n+1)^{d}}\rightarrow0\,.
$$
\end{proof}
In order to compute the norm of the finite volume approximation of the comb graphs $G\comb\bz^d$, we particularise \eqref{eksek3} by
putting 
\begin{equation}
\label{gnegne}
g_n(x):=\langle R_{A_{\bt_{2n+1}^d}}(x)\d_{\bf0},\d_{\bf0}\rangle^{-1}\,,\quad x\in(2d,+\infty)\,.  
\end{equation}
Even if it is not directly needed in the sequel, for the sake of completeness we provide some results which have a self--containing interest. 
Let $d\geq3$ and consider 
$$
f(\boldsymbol\th)=\frac{1}{\sum_{j=1}^d(1-\cos\th_j)}\,.
$$
together with its Fourier transform
$$
\hat f(\boldsymbol k)=\frac{1}{(2\pi)^d}\int_{\bt^d}\frac{e^{-\imath\langle {\bf k},{\boldsymbol\th}\rangle}}{\sum_{j=1}^d(1-\cos\th_j)}
\di^d{\boldsymbol\th}\,,
$$
which is meaningful as $f\in L^1(\bt^d,\di^d{\boldsymbol\th})$. We start by providing the following Tauberian theorem, probably known to the experts, concerning the  behaviour for $n\to +\infty$, of $\sum_{{\bf k}\in\S_n^d}|\hat f({\bf k})|^2$. The case 
$d>4$ is trivial because $f\in L^2(\bt^n)$, so we reduce the matter to the critical dimensions $d=3,4$.
\begin{Prop}
\label{astc}
For $n\rightarrow\infty$ we get
\begin{itemize}
\item[(i)] $\sum_{{\bf k}\in\S_n^d}|\hat f({\bf k})|^2\approx n$ for $d=3$,
\item[(ii)] $\sum_{{\bf k}\in\S_n^d}|\hat f({\bf k})|^2\approx\ln n$ for $d=4$.
\end{itemize}
\end{Prop}
\begin{proof}
Let $0\leq h\leq 1$ be a spherically symmetric positive cut--off  with small support which is  identically 1 around 
$0\in[-\pi,\pi]^d\sim\bt^d$. Let $g(\boldsymbol\th)=\frac{2}{\th^2}$
arising from the Taylor expansion of the cosine function around 0, with $\th=\sqrt{\sum_{j=1}^d\th_j^2}$.
It is immediate to see that
$$
f=gh+(f-g)h+f(1-h)\,.
$$
Put $F_1:=gh$, $F_2:=(f-g)h+f(1-h)$. As $F_2\in L^2(\bt^d,\di^d{\boldsymbol\th})$, by using Holder Inequality one shows that only
$\sum_{{\bf k}\in S_n^d}|\widehat F_1({\bf k})|^2$ contributes to the asymptotics of $\sum_{{\bf k}\in\S_n^d}|\hat f({\bf k})|^2$. 
Thus, we reduce the matter to that term and consider the function $gh$ as directly defined on the whole $\br^d$. Thus, 
$F_1(\boldsymbol\th)=s(r(\boldsymbol\th))$ for the function $s(r)=\frac1{r^2}$ on $\br_+$. By using the characteristic function of a spherically symmetric small neighbourhood of ${\bf0}$ as the cut--off, the Fourier transform of $F_1$ (considered as a function on the whole $\br^d$) can be expressed by the Hankel transform $S(\r)$ of $s(r)$ by
$$
S(\r)=2\pi\r^{1-\frac{d}2}\int_0^1r^{\frac{d}2-2}J_{\frac{d}2-1}(2\pi\r r)\di r\,.
$$
Here, $J_{\frac{d}2-1}$ is the Bessel function of the first kind of order $\frac{d}2-1$. By using Lemma 3.3 of \cite{Ja}, we get
$S(\r)\approx \r^{2-d}$ which leads to
$$
\sum_{{\bf k}\in\S_n^d}|\widehat F_1({\bf k})|^2\approx\int_1^n\r^{3-d}\di\r\,,
$$
and the proof follows.
\end{proof}
When $H=\bz^d$, it is possible to use the periodic boundary conditions. Thus, in this situation 
$\{G_n\comb\bt_{2n+1}^d\}_{n\in\bn}$ is an exhaustion of $G\comb\bz^d$, with periodic boundary conditions on the finer space, provided that $\{G_n\}_{n\in\bn}$ is an exhaustion of 
$G$. Following this line, we summarise some results of interest in the following
\begin{Prop}
\label{csette}
Consider any PF weight $w\in\cw_G$. The following assertions hold true.
\begin{itemize}
\item[(i)] If $\langle R_{A_{\bz^d}}(2d)\d_o,\d_o\rangle>\|A_G\|^{-1}$, the weight $v=w\otimes r$ on $G\comb\bz^d$ with
$\eps=\frac{\|A_{G\comb\bz^d}\|}2-d>0$ and
\begin{equation}
\label{errkapo}
r({\bf k}):=\frac{\|A_G\|}{2(2\pi)^d}\int_{\bt^d}\frac{e^{-\imath\langle {\bf k},{\boldsymbol\th}\rangle}}{\eps+\sum_{j=1}^d(1-\cos\th_j)}
\di^d{\boldsymbol\th}
\end{equation}
gives a PF weight for $A_{G\comb\bz^d}$. 
\item[(ii)] If $\langle R_{A_{\bz^d}}(2d)\d_o,\d_o\rangle\leq\|A_G\|^{-1}$, the weight $v=w\otimes r$ on $G\comb\bz^d$ with
\begin{equation}
\label{errkapo1}
r({\bf k}):=1+\frac{\|A_{G}\|}{2(2\pi)^d}\int_{\bt^d}\frac{e^{-\imath\langle {\bf k},{\boldsymbol\th}\rangle}-1}{\sum_{j=1}^d(1-\cos\th_j)}
\di^d{\boldsymbol\th}
\end{equation}
gives a PF weight for $A_{G\comb\bz^d}$. 
\item[(iii)] If $\langle R_{A_{\bz^d}}(2d)\d_o,\d_o\rangle>\|A_G\|^{-1}$ (which always happens when 
$d\leq2$), then $d_{PF}(G\comb\bz^d)=d_{PF}(G)$. If 
$\langle R_{A_{\bz^d}}(2d)\d_o,\d_o\rangle=\|A_G\|^{-1}$, then  $d_{PF}(G\comb\bz^d)=d_{PF}(G)+1$
when $d=3$, and $d_{PF}(G\comb\bz^d)=d_{PF}(G)$ when $d\geq4$ (apart of a logarithmically divergent term which does not contribute to $d_{PF}(G\comb\bz^4)$). If $\langle R_{A_{\bz^d}}(2d)\d_o,\d_o\rangle<\|A_G\|^{-1}$, then 
$d_{PF}(G\comb\bz^d)=d_{PF}(G)+d$.
\end{itemize}
\end{Prop}
\begin{proof}
Denote $\di m_n$, $\di m$ the normalized Haar measures on $\bt_{2n+1}^d$, $\bt^d$ respectively. Fix an exhaustion $\{G_n\}$ of $G$ together with a PF weight $w\in\cw_G$. Define $\eps_n>0$ as the unique solution of the Secular Equation
$$
\frac{\|A_{G_n}\|}2\int_{\bt_{2n+1}^d}\frac{\di m_n({\boldsymbol\th})}{\eps_n+\sum_{j=1}^d(1-\cos\th_j)}=1
$$
corresponding to $A_{G_n\comb \bt_{2n+1}^d}$. It is straightforward to see that 
$$
\lim_n\eps_n=\eps=\frac{\|A_{G\comb\bz^d}\|}2-d\geq0\,.
$$
We also get
$$
\frac1{\eps_n(2n+1)^d}=\frac2{\|A_{G_n}\|}-\int_{\bt_{2n+1}^d\backslash\{{\bf 0}\}}
\frac{\di m_n({\boldsymbol\th})}{\eps_n+\sum_{j=1}^d(1-\cos\th_j)}
$$
which leads to
\begin{equation*}
\lim_n\frac1{\eps_n(2n+1)^d}=\frac2{\|A_{G}\|}-\int_{\bt^d}
\frac{\di m({\boldsymbol\th})}{\eps+\sum_{j=1}^d(1-\cos\th_j)}
\end{equation*}
by Lemma 10.9 of \cite{FGI1}.

(i) and (ii) Fix $w\in\cw_G$ with the corresponding finite volume approximations $\{w_n\}$, all normalised at 1 on a common root 
of $G_n$.  Define $r_n$ on $\bt_{2n+1}^d$ as 
\begin{align*}
&r_n({\bf k})=\frac{\|A_{G_n}\|}2\int_{\bt_{2n+1}^d}
\frac{e^{-\imath\langle {\bf k},{\boldsymbol\th}\rangle}}{\eps_n+\sum_{j=1}^d(1-\cos\th_j)}\di m_n({\boldsymbol\th})\\
=&\frac{\|A_{G_n}\|}2\bigg(\frac1{\eps_n(2n+1)^d}
+\int_{\bt_{2n+1}^d\backslash\{{\bf 0}\}}
\frac{e^{-\imath\langle {\bf k},{\boldsymbol\th}\rangle}}{\eps_n+\sum_{j=1}^d(1-\cos\th_j)}\di m_n({\boldsymbol\th})\bigg)\,.
\end{align*}
As $\eps_n>0$, then $v_n=w_n\otimes r_n$ is the PF eigenvector on $G_n\comb\bt_{2n+1}^d$, unique with the given normalisation. In addition, $\lim_nr_n({\bf k})=r({\bf k})=:\langle r,\d_{\bf k}\rangle$. We also compute
\begin{align*}
&\langle A_{\bt_n^d}r_n,\d_{\bf k}\rangle=\frac{\|A_{G_n}\|}2\bigg(\frac{2d}{\eps_n(2n+1)^d}
+\int_{\bt_{2n+1}^d\backslash\{{\bf 0}\}}
\frac{2(\sum_{j=1}^d\cos\th_j)e^{-\imath\langle {\bf k},{\boldsymbol\th}\rangle}}{\eps_n+\sum_{j=1}^d(1-\cos\th_j)}\di m_n({\boldsymbol\th})\bigg)\\
&\rightarrow2d
+\|A_{G}\|\int_{\bt^d}
\frac{(\sum_{j=1}^d\cos\th_j)e^{-\imath\langle {\bf k},{\boldsymbol\th}\rangle}-d}{\eps_n+\sum_{j=1}^d(1-\cos\th_j)}
\di m({\boldsymbol\th})=:\langle A_{\bz^d}r,\d_{\bf k}\rangle
\end{align*}
By following the same lines of the proof of Propositon \ref{prcfi}, and taking into account the previous computations, we get in all the situations, 
\begin{align*}
&\langle A_{G\comb \bz^d}v,\d_g\otimes\d_{\bf k}\rangle
=\langle A_G w,\d_g\rangle\langle r,\d_{\bf 0}\rangle\d_{{\bf k},{\bf 0}}
+\langle w,\d_g\rangle\langle A_{\bz^d}r,\d_{\bf k}\rangle\\
=&\|A_G\|\langle w,\d_g\rangle\langle r,\d_{\bf 0}\rangle\d_{{\bf k},{\bf 0}}
+\langle w,\d_g\rangle\langle A_{\bz^d}r,\d_{\bf k}\rangle\\
=&\lim_n\big(\|A_{G_n}\|\langle w_n,\d_g\rangle\langle r_n,\d_{\bf 0}\rangle\d_{{\bf k},{\bf 0}}
+\langle w_n,\d_g\rangle\langle A_{\bt_{2n+1}^d}r_n,\d_{\bf k}\rangle\big)\\
=&\lim_n\big(\langle A_{G_n}w_n,\d_g\rangle\langle r_n,\d_{\bf 0}\rangle\d_{{\bf k},{\bf 0}}
+\langle w_n,\d_g\rangle\langle A_{\bt_{2n+1}^d}r_n,\d_{\bf k}\rangle\big)\\
=&\lim_n\langle A_{G_n\comb\bt_{2n+1}^d}v_n,\d_g\otimes\d_{\bf k}\rangle
=\lim_n\big(\|A_{G_n\comb\bt_{2n+1}^d}\|\langle v_n,\d_g\otimes\d_{\bf k}\rangle\big)\\
=&\lim_n\big(\|A_{G_n\comb\bt_{2n+1}^d}\|\langle w_n,d_g\rangle\langle r_n,\d_{\bf k}\rangle\big)
=\|A_{G\comb\bz^d}\|\langle w_n,\d_g\rangle\langle r_n,\d_{\bf k}\rangle\\
=&\|A_{G\comb\bz^d}\|\langle v,\d_g\otimes\d_{\bf k}\rangle\,.
\end{align*}

(iii) It follows directly by (i) and (ii), together with Proposition \ref{astc}.
\end{proof}

\section{The graph $\bn$}
\label{sec:en}

Contrarily to $\bz$, the Pure Hopping model on the graph $\bn$, together with the comb graphs whose base space is $\bn$ itself exhibits very interesting new phenomena concerning the appearance of the BEC. In order to study the spectral properties of the Adjacency, the first step is to point out the differences between the one--side chain corresponding to $\bn$, and the two--sides one corresponding to $\bz$. The starting point will be the infinite volume limit of finite volume approximations. The simple difference is to take the one--side segment
$$
S_n:=\{0,1,\dots,n\}\subset\bn
$$
made of $n+1$ points whose common root is the initial point $0$, as the finite volume approximations, and perform the one--side limit. The difference is to start for 
$\bz$ with the segment $\S_n=\{-n,-n+1,\dots,0,\dots,n-1,n\}$ made of $2n+1$ points whose common root is the middle point 0. In this case, the infinite volume limit will be a two--sides one, or a limit involving the periodic boundary conditions without essentially affecting the analysis. The finite volume approximation adopting the boundary conditions is evidently not allowed for $\bn$. Namely, the chosen exhaustion $\{\La_n\}_{n\in\bn}$ for $\bn$ will be $\La_n=S_n$
The explicit calculation concerning the spectral properties of such finite graphs $S_n$ as before are reported in \cite{VM} to which we refer the reader without further mention. 

The first main difference will concern the PF dimension which is easily seen to be $1$ for $\bz$, and $3$ for $\bn$ as it is summarised in the following result. Indeed, denote $v_n$ the PF eigenvector on $\S_n$, normalised at $1$ on $0\in\S_n$, and extended at $0$ on $\bn\backslash\S_n$.
\begin{Prop}
\label{pfene}
The weight $v(k):=k+1$ is a PF one for $A_\bn$. Then $d_{PF}(\bn)=3$. In addition, 
$$
\lim_nv_n(k)=v(k)\,,\quad j\in\bn\,,
$$
and
\begin{equation}
 \label{nn010}
\frac{\big\|v\lceil_{\La_n}\big\|^2}{\|v_n\|^2}=\frac{2\pi^2}3+o(1)\,.
\end{equation}
\end{Prop}
\begin{proof}
For $0\leq k,m\leq n$, the PF eigenvector, normalised to $1$ at $m$ is given by
$$
v_n(k)=\frac{\sin\frac{\pi(k+1)}{n+2}}{\sin\frac{\pi(m+1)}{n+2}}\longrightarrow\frac{k+1}{m+1}
$$
when $n\to+\infty$, and the first part follows by putting $m=0$. Concerning the PF dimension, we have
$$
\|v\lceil_{\La_n}\big\|^2=\sum_{k=0}^n(k+1)^2=\frac{(n+1)(n+2)(2n+3)}6\approx\frac{n^3}3\,,
$$
that is $d_{PF}(\bn)=3$. Finally, by taking into account the Riemann sum approximations of an integral of a continuos function, we compute
$$
\|v_n\|^2=\frac{n+1}{\sin^2\frac{\pi}{n+2}}\left(\frac1{n+1}\sum_{k=0}^n\sin^2\frac{\pi(k+1)}{n+2}\right)
\approx\bigg(\frac{n}{\pi}\bigg)^3\int_0^\pi\sin^2mx\di x=\frac{n^3}{2\pi^2}\,.
$$
Collecting together we get \eqref{nn010}\,.
\end{proof}
We pass to compute the matrix elements of the resolvent $R_{A_\bn}$, its limit when $\l\downarrow\|A_\bn\|$, together with its finite volume approximations. To shorten the notations in the various proofs, we put $A_n:=A_{S_n}$, $A:=A_\bn$, with the corresponding resolvent $R_n(\l)$, 
$R(\l)$, respectively. In addition $Q_n:=P^\perp_{v_n}$ is the orthogonal projection onto the orthogonal complement of the one--dimensional subspace generated by the PF eigenvectors for $A_{S_n}$.
Let $\l\geq 2$, define inductively
$$
\G_0(\l)=+\infty\,,\qquad \G_{n+1}(\l)=\l-\frac1{\G_{n}(\l)}\,,\,\,n=0,1,\dots\,.
$$
When $\l=2$, we easily compute
 \begin{equation*}
\G_n:=\G_n(2)=\frac{n+1}n\,,\,\,n=1,2,\dots\,.
\end{equation*}
\begin{Prop}
 \label{nn11}
If $\l>2$, we get
$$
\langle R_{A_\bn}(\l)\d_k,\d_{k+n}\rangle
=\frac2{\l-\frac2{\G_k(\l)}+\sqrt{\l^2-4}}\bigg(\frac{\l-\sqrt{\l^2-4}}2\bigg)^n\,,\quad k,n\in\bn\,.
$$
In addition, 
$$
\lim_{\l\downarrow2}\langle R(\l)\d_k,\d_l\rangle=(k\wedge l)+1\,.
$$
\end{Prop}
\begin{proof}
By taking into account the computation in Section 8 of \cite{FGI1}, we obtain 
$$
\begin{pmatrix} 
	\langle R(\l)\d_k,\d_{k+2n}\rangle\\
	 \langle R(\l)\d_k,\d_{k+2n+1}\rangle\\
 \end{pmatrix}
 :=\begin{pmatrix} 
	 \a_{n+1}\\
	 \b_{n+1}\\
 \end{pmatrix}
 =(\m_-)^n\begin{pmatrix} 
	 \a_{0}\\
	 \b_{0}\\
 \end{pmatrix}\,,
$$
where $\m_-=\frac{\l^2-2-\l\sqrt{\l^2-4}}2$. The initial value $\begin{pmatrix} 
	 \a_{0}\\
	 \b_{0}\\
 \end{pmatrix}$ can be determined by solving the finite system
\begin{align*}
&\l\s_1-\s_2=0\,,\\
&-\s_{i-1}+\l\s_i-\s_{i+1}=0\,,\,\, i=2,\dots,k-1\,,\\
&-\s_{k-1}+\l\s_k-\a_0=0\,,\\
&-\s_{k}+\l\a_0-\b_0=1\,,\\
&\begin{pmatrix} 
	 \a_{0}\\
	 \b_{0}\\
 \end{pmatrix}
 =a\begin{pmatrix} 
	2\\
	 \l-\sqrt{\l^2-4}\\
 \end{pmatrix}
 \end{align*}
 relative to the backward portion of the chain which leads to
$$
\langle R(\l)\d_k,\d_l\rangle=\frac2{\bigg[\l-\frac2{\G_k(\l)}+\sqrt{\l^2-4}\bigg]\prod_{m=l+1}^k\G_m(\l)}\,,\,\,l=0,1,\dots k-1\,.
$$
The last part follows by direct computation with $\l=2$, as the r.h.s. of the formulae defining 
 $\langle R(\l)\d_k,\d_l\rangle$ are continuous for $\l\downarrow\|A_\bn\|=2$.
\end{proof}
We pass to the investigation of the thermodynamic limit, starting with that concerning the two--point function.
\begin{Prop}
\label{pefo}
Let $\l_n\geq\|A_{S_n}\|$ such that $\lim_n\l_n=2=\|A_\bn\|$. Then
$$
\lim_n\langle R_{S_n}(\l_n)Q_n\d_k,Q_n\d_l\rangle=\langle R_{A_\bn}(2)\d_k,\d_l\rangle\,.
$$
\end{Prop}
\begin{proof}
We start by noticing that 
\begin{align}
\label{ligez}
&\frac1{\pi}\int_0^\pi\sin^2x\di x
\leq\frac1{n+2}\sum_{k=1}^{n+1}\sin^2\frac{\pi k}{n+2}\longrightarrow\frac1{\pi}\int_0^\pi\sin^2x\di x=\frac12\nn\,,\\
&\frac1{n+2}\sum_{k=1}^{n+1}\sin^2\frac{\pi mk}{n+2}\longrightarrow\frac1{\pi}\int_0^\pi\sin^2mx\di x=\frac12\,,\quad m=1,2,\dots\,.
\end{align}
Define 
$$
\n_n:=\l_n-\|A_n\|=\l_n-2\cos\frac{\pi}{n+2}\,.
$$
By using the results in Section 5.4 of \cite{VM}, we compute after a bit of calculations involving trigonometric functions, 
\begin{align*}
\langle R_n&(\l_n)Q_n\d_k,Q_n\d_l\rangle\leq\frac1{2(n+1)}
\sum_{m=1}^{n}\frac{\big|\sin\frac{\pi(m+1)(k+1)}{n+2}\sin\frac{\pi(m+1)(l+1)}{n+2}\big|}
{\sin^2\frac{\pi m}{2(n+2)}}\\
\leq&\frac1{2(n+1)}\sum_{m=1}^{n}\bigg(
\frac{\big|\sin\frac{\pi m(k+1)}{n+2}\sin\frac{\pi m(l+1)}{n+2}\big|+
\sin\frac{\pi(k+1)}{n+2}\big|\sin\frac{\pi m(l+1)}{n+2}\big|}{\sin^2\frac{\pi m}{2(n+2)}}\\
&\quad\quad\quad\quad\quad\quad+\frac{\sin\frac{\pi(k+1)}{n+2}\big|\sin\frac{\pi m(l+1)}{n+2}\big|+
\sin\frac{\pi(k+1)}{n+2}\sin\frac{\pi(l+1)}{n+2}}
{\sin^2\frac{\pi m}{2(n+2)}}\bigg)\,.
\end{align*}
Concerning the last three addenda, we get
\begin{align*}
\frac1{2(n+1)}\sum_{m=1}^{n}
&\frac{\sin\frac{\pi(k+1)}{n+2}\big|\sin\frac{\pi m(l+1)}{n+2}\big|
+\sin\frac{\pi(k+1)}{n+2}\big|\sin\frac{\pi m(l+1)}{n+2}\big|+
\sin\frac{\pi(k+1)}{n+2}\sin\frac{\pi(l+1)}{n+2}}
{\sin^2\frac{\pi m}{2(n+2)}}\\
=&O\bigg(\frac1{n^2}\sum_{m=1}^{n}\frac{\sin\frac{\pi m}{n+2}}{\sin^2\frac{\pi m}{2(n+2)}}\bigg)
\leq O\bigg(\frac1{n}\int_{\frac1n}^\pi\frac{\sin x}{\sin^2\frac x2}\di x\bigg)\longrightarrow 0\,.
\end{align*}
Thus, by using a generalised version of Lebesgue Dominated Convergence Theorem (cf. Theorem 19 in Section 4.4 of \cite{RF})
and retaining the leading terms, we get by taking into account of \eqref{ligez},
\begin{align*}
&\langle R_n(\l_n)Q_n\d_k,Q_n\d_l\rangle=\sum_{m=2}^{n+1}\frac{\sin\frac{\pi m(k+1)}{n+2}\sin\frac{\pi m(l+1)}{n+2}}
{\n_n+2\big(\cos\frac{\pi}{n+2}-\cos\frac{\pi m}{n+2}\big)\sum_{k=1}^{n+1}\sin^2\frac{\pi mk}{n+2}}\\
\approx&\frac1{n+1}\sum_{m=2}^{n+1}\frac{\sin\frac{\pi m(k+1)}{n+2}\sin\frac{\pi m(l+1)}{n+2}}
{\cos\frac{\pi}{n+2}-\cos\frac{\pi m}{n+2}}
=\frac1{2(n+1)}\sum_{m=1}^{n}\frac{\sin\frac{\pi(m+1)(k+1)}{n+2}\sin\frac{\pi(m+1)(l+1)}{n+2}}
{\sin\frac{\pi(m+2)}{2(n+2)}\sin\frac{\pi m}{2(n+2)}}\\
\approx&\frac1{2(n+1)}\sum_{m=1}^{n}\frac{\sin\frac{\pi m(k+1)}{n+2}\sin\frac{\pi m(l+1)}{n+2}}
{\sin^2\frac{\pi m}{2(n+2)}}\longrightarrow\frac1{2\pi}\int_0^\pi\frac{\sin(k+1)x\sin(l+1)x}{\sin^2\frac x2}\di x\,,
\end{align*}
where the last step is justified as the integrand is continuous in $[0,\pi]$. Concerning the last integral and the case $k\neq l$, we compute by taking into account the  definition and the properties of Fej\'er kernel $\F_n$ (see e.g. \cite{KF}),
\begin{align*}
&\frac1{2\pi}\int_0^\pi\frac{\sin(k+1)x\sin(l+1)x}{\sin^2\frac x2}\di x
=\frac1{4\pi}\int_{-\pi}^\pi\frac{\cos|k-l|x-\cos(k+l+2)x}{2\sin^2\frac x2}\di x\\
=&\frac1{4\pi}\int_{-\pi}^\pi\frac{\sin^2(k+l+2)\frac x2}{\sin^2\frac x2}\di x-\frac1{4\pi}\int_{-\pi}^\pi\frac{\sin^2|k-l|\frac x2}
{\sin^2\frac x2}\di x\\
=&\frac{k+l+2}2\int_{-\pi}^\pi\F_{k+l+2}(x)\di x-\frac{|k-l|}2\int_{-\pi}^\pi\F_{|k-l|}(x)\di x\\
=&(k\wedge l)+1\,.
\end{align*}
Similarly, the case $k=l$ leads to
$$
\frac1{2\pi}\int_0^\pi\frac{\sin^2(k+1)x}{\sin^2\frac x2}\di x=(k+1)\int_{-\pi}^\pi\F_{2(k+1)}(x)\di x=k+1\,.
$$
\end{proof}
Now we pass to the infinite volume limit involving the density of particles. To do that we start with the following
\begin{Lemma}
\label{starct}
Let $\{b_n\}_{n\in\bn}\subset\br_+$ be a sequence of positive number such that $\lim_nb_n=0$. Then
$$
\lim_n\int_{\frac1{n}}^{1}\frac{\di x}{b_n+x^2}=+\infty\,.
$$
\end{Lemma}
\begin{proof}
\begin{align*}
&\int_{\frac1{n}}^{1}\frac{\di x}{b_n+x^2}=\frac1{\sqrt{b_n}}\bigg(\arctan\frac1{\sqrt{b_n}}-\arctan\frac1{n\sqrt{b_n}}\bigg)\\
=&\frac1{\sqrt{b_n}}\arctan\frac{\frac1{\sqrt{b_n}}\big(1-\frac1n\big)}{1+\frac1{nb_n}}
\approx\frac1{\sqrt{b_n}}\arctan\frac1{\sqrt{b_n}+\frac1{n\sqrt{b_n}}}\,.
\end{align*}
Consider any subsequence $\{b_{n_k}\}_{k\in\bn}$ such that $n_k\sqrt{b_{n_k}}\to c$. If $c\in(0,+\infty)\cup\{+\infty\}$ then 
$\arctan\frac1{\sqrt{b_{n_k}}+\frac1{n_k\sqrt{b_{n_k}}}}\to K>0$. This means that 
$\int_{\frac1{n_k}}^{1}\frac{\di x}{b_{n_k}+x^2}\to+\infty$. If $c=0$ then
$$
\arctan\frac1{\sqrt{b_{n_k}}+\frac1{n_k\sqrt{b_{n_k}}}}\approx
\frac1{\sqrt{b_{n_k}}+\frac1{n_k\sqrt{b_{n_k}}}}\,.
$$
Thus, again
$$
\int_{\frac1{n_k}}^{1}\frac{\di x}{b_{n_k}+x^2}\approx
\frac1{b_{n_k}+\frac1{n_k}}\longrightarrow+\infty
$$
and the assertion follows.
\end{proof}
\begin{Prop}
\label{traccaza}
Let $\l_n>\|A_n\|$ such that $\lim_n\l_n=2=\|A\|$. Then
$$
\lim_n\t_n(Q_nR_{S_n}(\l_n)Q_n)=+\infty\,.
$$
\end{Prop}
\begin{proof}
Again by using the explicit computations in in Section 5.4 of \cite{VM}, and for $\n_n$ as before with $b_n=\n_n/8$, we get
\begin{align*}
\t_n&(Q_nR_n(\l_n)Q_n)=\frac1{n+1}\sum_{m=2}^{n+1}\frac1{\n_n+2\big(\cos\frac\pi{n+2}-\cos\frac{\pi m}{n+2}\big)}\\
=&\frac1{n+1}\sum_{m=1}^{n}\frac1{\n_n+4\big(\cos\frac\pi{n+2}\sin^2\frac{\pi m}{2(n+2)}
+\sin\frac\pi{n+2}\sin\frac{\pi m}{2(n+2)}\cos\frac{\pi m}{2(n+2)}\big)}\\
\geq&\frac1{(n+1)}\sum_{m=2}^{n}\frac1{b_n+\sin^2\frac{\pi m}{2(n+2)}}
\geq\frac1{4\pi}\int_{\frac\pi{(n+2)}}^{\frac{(n+1)\pi}{2(n+2)}}\frac{\di x}{b_n+\sin^2 x}\\
&\geq\frac1{4\pi}\int_{\frac\pi{(n+2)}}^{\frac{(n+1)\pi}{2(n+2)}}\frac{\di x}{b_n+x^2}
\approx\frac1{4\pi}\int_{\frac\pi{(n+2)}}^{\pi}\frac{\di x}{b_n+x^2}\rightarrow+\infty
\end{align*}
by Lemma \ref{starct}.
\end{proof}
Concerning the thermodynamic limit for the Adjacency of the graph $\bn$, the result in Proposition \ref{traccaza} is in accordance to the fact that the critical density for the Pure Hopping model on $\bn$ is infinite (cf. Remark 8.4 of \cite{FGI1}), whereas Proposition \ref{pefo} is compatible with the fact that the Adjacency $A_\bn$ is transient. The latter property is the necessary and sufficient condition for the existence of locally normal states exhibiting BEC, see Theorem \ref{tretre}. The next theorem summarise the previous results and describes in details how it is possible to construct the states \eqref{kmsbecd} by infinite volume limit of the Bose--Gibbs grand--canonical finite volume ensemble.
\begin{Thm}
\label{liift}
Let $D\geq0$, and $v(k):=k+1$, $k\in\bn$, be the PF weight of $A_\bn$ together with the finite dimensional PF eigenvectors 
$\{v_n\}_{n\in\bn}$ for $A_{S_n}$, all normalized at 1 on the common root 0. For each sequence of the chemical potentials
$\{\m_n\}_{n\in\bn}$ with $\m_n<\|A_\bn\|-\|A_{S_n}\|$, such that $\lim_n\m_n=0$ and 
$$
\lim_{\m_n\to0}\frac1{\|v_n\|^2(\|A_\bn\|-\|A_{S_n}\|-\m_n)}=D\,,
$$
we get 
\begin{align*}
&\lim_n\left\langle\left(e^{(\|A_\bn\|-\m_n)\idd_{\ell^2(S_n)}-A_{S_n}}-\idd_{\ell^2(S_n)}\right)^{-1}
P_{\ell^2(S_n)}u_1,P_{\ell^2(S_n)}u_2\right\rangle\\
=&\left\langle\left(e^{\|A_\bn\|\idd-A_\bn}-\idd\right)^{-1}u_1,u_2\right\rangle+D\langle u_1,v\rangle\langle v,u_2\rangle\quad u_1,u_2\in\gph\,.
\end{align*}
\end{Thm}
\begin{proof}
Thanks to the functional calculus with the function in \ref{resam}, we can reduce the matter to the resolvent. 
Put $u_i=e^{\imath t_i H}\d_{x_i}$, $i=1,2$, and $R^{(n)}:=Q_nR_n(\l_n)Q_n$ which is well defined as $\l_n=\|A\|-\m_n>\|A_n\|$. By following the same lines of 
Proposition \ref{pefo}, and taking into account \eqref{ligez}, we compute for the four terms appearing in the matrix elements $R^{(n)}_{jk}$,
\begin{align*}
&R^{(n)}_{jk}\leq\frac3{4\pi}\text{length}[0,\pi]
\bigg(\max_{x\in[0,\pi]}\frac{|\sin(j+1)x\sin(k+1)x|}{\sin^2\frac x2}\\
+\max_{x\in[0,\pi]}&\bigg|\frac{\sin(j+1)x}{\sin\frac x2}\bigg|\sup_{n\geq1}\bigg|\frac{\sin\frac{\pi(k+1)}{n+2}}
{\sin\frac{\pi}{2(n+2)}}\bigg|
+\max_{x\in[0,\pi]}\bigg|\frac{\sin(k+1)x}{\sin\frac x2}\bigg|\sup_{n\geq1}\bigg|\frac{\sin\frac{\pi(l+1)}{n+2}}
{\sin\frac{\pi}{2(n+2)}}\bigg|\bigg)\\
+&\frac12\sup_{n\geq1}\frac{\big|\sin\frac{\pi(l+1)}{n+2}\sin\frac{\pi(k+1)}{n+2}\big|}
{\sin^2\frac{\pi}{2(n+2)}}\leq11(j+1)(k+1)\,.
\end{align*}
Suppose that $|z|=|w|=r>2$. Proposition \ref{nn11} says that $R(r)_{jk}$ is exponentially decreasing for $j$ or $k$ going to infinity, then we get
$$
\bigg|\sum_{j,k=0}^n\overline{R(z)_{ji}}R^{(n)}_{jk}R(w)_{kl}\bigg|\leq11\bigg(\sum_{j=0}^{+\infty}R(r)_{ji}(j+1)\bigg)
\bigg(\sum_{k=0}^{+\infty}R(r)_{kl}(k+1)\bigg)
<+\infty\,.
$$
Now, by using Lebesgue Dominated Convergence Theorem, Proposition 4.2 and Theorem \ref{tretre}, we obtain
\begin{align*}
&\lim_n\langle R^{(n)}u_1,u_2\rangle\\
=-&\lim_n\frac{1}{4\pi^2}\oint_{C_r}\di ze^{-\imath t_1(\|A\|-\bar z)}\oint_{C_r} \di w e^{\imath t_2(\|A\|-w)}
\sum_{j,k=0}^n\overline{R(z)_{jx_1}}R^{(n)}_{jk}R(w)_{kx_2}\\
=-&\frac{1}{4\pi^2}\oint_{C_r}\di ze^{-\imath t_1(\|A\|-\bar z)}\oint_{C_r} \di w e^{\imath t_2(\|A\|-w)}
\sum_{j,k=0}^{+\infty}\overline{R(z)_{jx_1}}\big(\lim_nR^{(n)}_{jk}\big)R(w)_{kx_2}\\
=-&\frac{1}{4\pi^2}\oint_{C_r}\di ze^{-\imath t_1(\|A\|-\bar z)}\oint_{C_r} \di w e^{\imath t_2(\|A\|-w)}
\sum_{j,k=0}^{+\infty}\overline{R(z)_{jx_1}}R(2)R(w)_{kx_2}\\
=&\langle R(2)u_1,u_2\rangle\,.
\end{align*}
Concerning the condensate portion, with the obvious notations we first note that, under the condition $\m_n\to0$, it reduces to
$$
\left\langle\left(e^{(\|A\|-\m_n)\idd_n-A_n}-\idd_n\right)^{-1}P_{v_n}u_1,P_{v_n}u_2\right\rangle
\approx\frac{\langle u_1,v_n\rangle\langle v_n,u_2\rangle}{\|v_n\|^2(\|A\|-\|A_n\|-\m_n)}\,.
$$
Thus, the matter is reduced to investigate the limit of the numerator. By
$$
v_n(k)=\frac{\sin\frac{\pi(k+1)}{n+2}}{\sin\frac{\pi}{n+2}}\leq k+1\,,\quad k=0,1,\dots,n\,,
$$
and reasoning as before, we obtain
$$
\sum_{k=0}^nR(r)_{jk}v_n(k)\leq\sum_{k=0}^{+\infty}R(r)_{jk}(k+1)<+\infty\,.
$$
Thus, by Proposition \ref{pfene} and Theorem \ref{tretre}, for $u=e^{\imath tH}\d_x\in\gph$ we get,
\begin{align*}
\lim_n\langle v_n,u\rangle
=&\lim_n\frac{1}{2\pi\imath}\oint_{C_r}\di ze^{\imath t(\|A\|-z)}
\sum_{k=0}^nR(z)_{xk}v_n(k)\\
=&\frac{1}{2\pi\imath}\oint_{C_r}\di ze^{\imath t(\|A\|-z)}
\sum_{k=0}^{+\infty}R(z)_{xk}(\lim_nv_n(k))\\
=&\frac{1}{2\pi\imath}\oint_{C_r}\di ze^{\imath t(\|A\|-z)}
\sum_{k=0}^{+\infty}R(z)_{xk}v(k)\\
=&\langle v,u\rangle\,.
\end{align*}
Collecting together the results concerning the excited states and the portion of the condensate, the proof follows.
\end{proof}
\begin{Rem}
\label{reenai}${}$\\
\noindent
{\bf(i)} The graph $\bn$ is transient, so it exhibit locally normal states describing BEC, see Theorem \ref{tretre}. The finite volume two--point function splits into two terms. In the infinite volume limit, the first one converges to the corresponding term of the two--point function given in the l.h.s. of \eqref{kmsbecd}, the former and the latter describing the occupation portion of the excited levels for the finite and infinite volume, respectively. The second one, describing the occupation portion of the ground state, converges to the last addendum in the l.h.s. of \eqref{kmsbecd} corresponding to the portion of the condensate. This is precisely Theorem \ref{liift}.\\
\vskip.1cm
\noindent
{\bf(ii)} Concerning the mean density, as before two parts contribute to that: the part taking into account the excited levels and the condensation portion, respectively. 
Being $\r_c=+\infty$ and considering the functional calculus by the bounded function in \eqref{resam}, the first part always diverges in the infinite volume limit by giving
\begin{align*}
\lim_n\r\big(\om_{\La_n}\big)=+\infty=\r_c
=&\lim_n\frac1{|\La_n|}\sum_{x\in\La_n}\langle f(\|A\|\idd-A)\d_x,\d_x\rangle\\ 
+&\lim_n\frac1{|\La_n|}\sum_{x\in\La_n}\langle R_A(\|A\|)\d_x,\d_x\rangle\,.
\end{align*}
For each state $\om=\om_D$ in \eqref{kmsbecd} describing the condensation regime, its mean density $\r(\om)=+\infty$ even if they are locally normal, i.e. the local density is finite:
$$
\r_\La(\om)=\sum_{x\in\La}\om\big(a^\dagger(\d_x)a(\d_x)\big)<+\infty\,.
$$
Concerning the second terms, we have to distinguish two cases arising from the condensation regime.\\
\vskip.1cm
\noindent
{\bf(iii)} Suppose in \eqref{kmsbecd} that $D>0$. Being $3=d_{PF}>d_G=1$,
\begin{align*}
\lim_n\r^{{\text cond}}\big(\om_{\La_n}\big)=&\lim_n\frac{\|v_n\|^2}{|\La_n|}\lim_n\frac1{\|v_n\|^2(\|A\|-\|A_n\|-\m_n)}=+\infty\\
=&D\lim_n\frac{\sum_{x\in\La_n}v(x)^2}{|\La_n|}=\lim_n\r_{\La_n}^{{\text cond}}(\om)\,.
\end{align*}
But the ratio between the mean density of the condensate and its infinite volume limit is different from 1 by \eqref{nn010}. In fact
$$
\frac{\r_{\La_n}^{{\text cond}}(\om)}{\r^{{\text cond}}\big(\om_{\La_n}\big)}
=\frac{\big\|v\lceil_{\La_n}\big\|^2}{\|v_n\|^2}\longrightarrow\frac{2\pi^2}3\,.
$$
\vskip.1cm
\noindent
{\bf(iv)} The critical case $D=0$ leads to the following situations:
\begin{align*}
\r^{{\text cond}}&\big(\om_{\La_n}\big)=\frac{\|v_n\|^2}{|\La_n|}\frac1{\|v_n\|^2(\|A\|-\|A_n\|-\m_n)}\\
=&\bigg(\frac1{2\pi^2}+o(1)\bigg)\frac{(n+1)^2}{\|v_n\|^2(\|A\|-\|A_n\|-\m_n)}\,,
\end{align*}
whose possible limit depends on the rate of the convergence to 0 of $\frac1{\|v_n\|^2(\|A\|-\|A_n\|-\m_n)}$.
\end{Rem}

\section{The comb graph $\bn\comb\bz^d$}
\label{snd}

The results in the previous sections allow us to investigate in the full generality the surprising phenomena described below, relative to the appearance of the BEC for the Pure Hopping model on the comb graphs $\bn\comb\bz^d$.
As we have shown in Proposition \ref{pradneg}, the comb graph $N_d:=\bn\comb\bz^d$ is a negligible additive perturbation of the (non connected) graph consisting of the disjoint union of $\bn$ copies of $\bz^d$. The first step is to decide whether its Adjacency admits Hidden Spectrum. 
\begin{Prop}
\label{hspcn}
The Adjacency of the comb graph $N_d$ has Hidden Spectrum if and only if
$d\in\{1,2\}$.
\end{Prop}
\begin{proof}
As $A_{\bz^d}$ does not have Hidden Spectrum, $A_{N_d}$ has Hidden Spectrum if and only if $\|A_{N_d}\|>\|A_{\bz^d}\|$.
It happens if and only if the Secular Equation (cf. \eqref{eksek1})
\begin{equation}
\label{sigek}
\langle R_{A_{\bz^d}}(\l)\d_0,\d_0\rangle\|A_\bn\|=1\,,
\end{equation}
has a (necessarily unique) solution $\l_*>\|A_{\bz^d}\|=2d$. This can happen
if and only if
$$
\lim_{\l\downarrow2d}\langle R_{A_{\bz^d}}(\l)\d_0,\d_0\rangle>1/2\,.
$$
Note that
 \begin{align}
 \label{nngr00}
\langle R_{A_{\bz^{d+1}}}(&2(d+1))\d_0,\d_0\rangle
=\frac1{\pi^{d+1}}\int\!\!\!\!\int\!\!\dots\!\!\int_{[0,\pi]^{d+1}}\frac{\bf\di^{d+1}\th}{\sum_{k=1}^{d+1}2(1-\cos\th_k)}\nn\\
\leq&\frac1{\pi^d}\int\!\!\!\!\int\!\!\dots\!\!\int_{[0,\pi]^d}\frac{\bf\di^d\th}{\sum_{k=1}^{d}2(1-\cos\th_k)}
=\langle R_{A_{\bz^{d}}}(2d)\d_0,\d_0\rangle
\end{align}
If $d=1,2$, by \eqref{nngr00} and using the polar coordinates we get,
 \begin{align*}
\langle R_{A_{\bz}}(2)\d_0,\d_0\rangle\geq\langle R_{A_{\bz^2}}(4)\d_0,\d_0\rangle
=&\frac1{\pi^2}\int\!\!\!\!\int_{[0,\pi]^2}\frac{\di\th_1\di\th_2}{2(1-\cos\th_1)+2(1-\cos\th_2)}\\
\geq&\frac1{\pi^2}\int\!\!\!\!\int_{[0,\pi]^2}\frac{\di\th_1\di\th_2}{\th_1^2+\th_2^2}
\geq\frac1{2\pi}\int_0^\pi\frac{\di r}r=+\infty
\end{align*}
Thus, if $d=1,2$, $A_{N_d}$ has Hidden Spectrum. Concerning $d\geq3$, it follows by (2.1) and (2.8) of \cite{J} that (with $P(1)$ defined there)
$$
2\langle R_{A_{\bz^3}}(6)\d_0,\d_0\rangle=P(1)/3<0.6<1\,.
$$
Namely, the Secular Equation \eqref{sigek} does not have any solution for $\l>6$ when $d=3$, that is $N_3$ does not have Hidden Spectrum. By \eqref{nngr00}, we conclude that $N_d$ cannot have Hidden Spectrum for each $d\geq3$.
\end{proof}
Fix $\La_n:=S_n\comb\bt_{2n+1}^d$, together with the PF eigenvector $w_n$ on the segment $S_n=\{0,1,\dots,n\}$ normalised to 1 on the common root 0. Let $L_n$ be the unique solution of the Secular Equation
$$
\langle R_{A_{\bt_{2n+1}^d}}(\l_n)\d_0,\d_0\rangle\|A_{S_n}\|=1\,.
$$
By Lemma \ref{caficz}, $L_n>2d$. In addition, by Proposition \ref{hspcn} $L_n\to\|A_{N_d}\|$, where $\|A_{N_d}\|\geq2d$ either satisfies \eqref{sigek} for $d=1,2$, or
$\|A_{N_d}\|=2d$ if $d\geq3$. Put
$\eps_n:=\frac{L_n-2d}2$. Then 
$$
\eps_n\to\eps:=\frac{\|A_{N_d}\|-2d}2\geq0\,.
$$
The PF eigenvector on the Comb $\La_n$ is then given by $v_n:=w_n\otimes R_{A_{\bt_{2n+1}^d}}(L_n)\d_0$. As usual, we extend the $v_n$ to all $\ell^2(VN_d)$ by putting 0 elsewhere in $N_d\backslash\La_n$. 
By Propositions
\ref{prcfi} and \ref{pfene}, $v_n$ converges point--wise to $v=w\otimes r$, where $w(j)=j+1$ is the PF weight on $\bn$ (which can be proved to be unique up to multiplicative costants), and $r({\bf k})$ is given in \eqref{errkapo} if $d=1,2$, and \eqref{errkapo1} if 
$d\geq3$. We specialise the results about the PF weight in the following
\begin{Prop}
\label{pfenenz}
For the PF dimension, we have $d_{PF}(N_d)=3$ if $d=1,2$, and $d_{PF}(N_d)=3+d$ if $d\geq3$. For $d=1,2$, \eqref{nn010}
$$
\lim_n\frac{\big\|v\lceil_{\La_n}\big\|^2}{\|v_n\|^2}=\frac{2\pi^2}3
$$
holds true as well.
\end{Prop}
\begin{proof}
The proof of the first part follows collecting together Propositions \ref{hspcn}, \ref{prcfi} and \ref{pfene}. Concerning the ratio 
$\frac{\big\|v\lceil_{\La_n}\big\|^2}{\|v_n\|^2}$, for $d=1,2$, $\|A_{N_d}\|>\|A_{\bz^d}\|$, which implies 
$$
R_{A_{\bt_{2n+1}^d}}(\l_n)\d_{\bf k}\longrightarrow R_{A_{\bz^d}}(\|A_{N_d}\|)\d_{\bf k}\,,\quad {\bf k}\in\bz^d\,.
$$
Then we get 
$$
\frac{\big\|v\lceil_{\La_n}\big\|^2}{\|v_n\|^2}=\frac{\big\|w\lceil_{[0,n]}\big\|^2}{\|w_n\|^2}
\frac{\|R_{A_{\bt_{2n+1}^d}}(\l_n)\d_{\bf0}\|^2}{\|R_{A_{\bz^d}}(\|A_{N_d}\|)\d_{\bf0}\|^2}\longrightarrow\frac{2\pi^2}3\,.
$$
\end{proof}
Due to the appearance of the Hidden Spectrum, from now in in this section if it is not otherwise specified, we limit the analysis to the cases $d=1,2$, more interesting for the investigation of the BEC due to the inhomogeneity. In this case, the PF weight is written as
\begin{equation}
\label{fifi}
v(j,{\bf k})=\frac{j+1}{(2\pi)^d}\int_{\bt^d}\frac{e^{\imath\langle {\bf k},{\bf\th}\rangle}}{\eps+\sum_{j=1}^d(1-\cos\th_j)}
\di^d{\bf\th}\,,\quad (j,{\bf k})\in \bn\times\bz^d\,.
\end{equation}
As in the previous section, we pass to the infinite volume limit of the finite approximations of the two--point function w.r.t the exhaustion given by $\La_n=S_n\comb\bt_{2n+1}$. 
Consider the self--adjoint projection $Q_n$ onto the codimension $1$ orthogonal subspace to the finite volume PF eigenvector $v_n=w_n\otimes r_n$.
\begin{Prop}
\label{cocaszon}
Fix any sequence $\{\l_n\}_{n\in\bn}$ such that $\l_n>\|A_{\La_n}\|$ and 
$\lim_n\l_n=\|A_{\bn\comb\bz^d}\|$. With $\eps=\frac{\|A_{N_d}\|}2-d$ and $\d_{-1}:=0$ in $\ell^2(\bn)$, we get
\begin{align*}
&\lim_n\langle R_{A_{\La_n}}(\l_n)Q_n\d_k\otimes\d_{\bf m},Q_n\d_l\otimes\d_{\bf n}\rangle\\
=\frac{\d_{kl}}{2(2\pi)^d}\int_{\bt^d}
&\frac{e^{\imath\langle {\bf m}-{\bf n},{\boldsymbol\th}\rangle}}{\eps+\sum_{j=1}^d(1-\cos\th_j)}
\di^d{\boldsymbol\th}+\frac{\langle R_{A_{\bn}}(2)(\d_{k-1}+\d_{k+1}),\d_l\rangle}{2(2\pi)^{2d}}\\
\times\int\!\!\!\!\int_{\bt^d\times\bt^d}
&\frac{e^{\imath(\langle {\bf m},{\boldsymbol\a}\rangle-\langle {\bf n},{\boldsymbol\b}\rangle)}}
{(\eps+\sum_{j=1}^d(1-\cos\a_j)(\eps+\sum_{j=1}^d(1-\cos\b_j))}
\di^d{\boldsymbol\a}
\di^d{\boldsymbol\b}\\
=&\lim_{\l\downarrow\|A_{N_d}\|}\langle R_{A_{N_d}}(\l)\d_k\otimes\d_{\bf m},\d_l\otimes\d_{\bf n}\rangle\,.
\end{align*} 
Thus, $A_{\bn\comb\bz^d}$ is transient.
\end{Prop}
\begin{proof}
We suppose that the PF eigenvector $w_n$ of $A_{S_n}$ is normalised such that $\|w_n\|=1$. Then $w_n(i)\to0$, $i\in\bn$. Put $L_n:=\|\La_n\|<\l_n$, $\eps_n:=\frac{\l_n}2-d$, then 
$\eps_n\to\eps>0$. Taking into account \eqref{nn0110}, we have with the obvious notations, 
$$
Q_n=\idd\otimes\idd-P_{w_n}\otimes P_{r_n}=P_{w_n}\otimes(\idd-P_{r_n})+P^\perp_{w_n}\otimes \idd\,.
$$
In addition, denoting $\R_n(\l)$ as the resolvent of $A_{\bt_{2n+1}^d}$, we get
$$
P_{r_n}=\frac{\langle\,{\bf\cdot}\,, \R_n(L_n)\d_{\bf 0}\rangle}{\|\R_n(L_n)\d_{\bf 0}\|^2}\R\!{}_n(L_n)\d_{\bf 0}\,.
$$
By orthogonality, the unique surviving terms are
\begin{align}
\label{secsct}
Q_nR_{A_{\La_n}}(\l_n)Q_n=&P_{w_n}\otimes(\idd-P_{r_n})R_{A_{\La_n}}(\l_n)P_{w_n}\otimes(\idd-P_{r_n})\nn\\
+&P^\perp_{w_n}\otimes \idd R_{A_{\La_n}}(\l_n)P^\perp_{w_n}\otimes \idd\,.
\end{align}
Concerning the first one, we note that it can be written in the form
\begin{align}
\label{secsct1}
&\langle P_{w_n}\otimes(\idd-P_{r_n})R_{A_{\La_n}}(\l_n)P_{w_n}\otimes(\idd-P_{r_n})\d_k\otimes\d_{\bf m},\d_l\otimes\d_{\bf n}\rangle\\
=&w_n(k)w_n(l)\bigg\langle\bigg((\idd-P_{r_n})\R\!{}_n(\l_n)(\idd-P_{r_n})+\frac{g_n(\l_n)g_n(L_n)}{g_n(\l_n)-g_n(L_n)}A_n(\l_n)\bigg)\d_{\bf m},\d_{\bf n}\bigg\rangle\,,\nn
\end{align}
where
$$
A_n(\l_n):=(\idd-P_{r_n})\R\!{}_n(\l_n)P_{0}\R\!{}_n(\l_n)(\idd-P_{r_n})\,,
$$
and $g_n$ is given by \eqref{gnegne}.
Notice that $A_n(L_n)=0$. By using the Taylor expansion, we get
\begin{align*}
&\frac{\langle A_n(\l_n)\d_{\bf m},\d_{\bf n}\rangle}{g(\l_n)-g(L_n)}
=-\frac{\langle\R\!{}_n(\s_n)\d_{\bf 0},\d_{\bf 0}\rangle^2}
{\langle\R\!{}_n(\s_n)^2\d_{\bf 0},\d_{\bf 0}\rangle}\\
\times\langle(\idd-P_{r_n})&(\R\!{}_n(\eta_n)^2P_0\R\!{}_n(\eta_n)
+\R\!{}_n(\eta_n)P_0\R\!{}_n(\eta_n)^2)(\idd-P_{r_n})\d_{\bf m},\d_{\bf n}\rangle\,,
\end{align*}
where $\s_n,\eta_n\in(L_n,\l_n)$, $n=1,2,\dots$\,. Thus, \eqref{secsct1} goes to zero as the coefficient of $w_n(k)w_n(l)$ in the l.h.s. is bounded.

Concerning the second addendum in \eqref{secsct}, we get
\begin{align*}
\lim_n\int_{\bt_{2n+1}^d}
\frac{e^{\imath\langle {\bf m},{\boldsymbol\th}\rangle}}{\eps_n+\sum_{j=1}^d(1-\cos\th_j)}
\di m_n({\boldsymbol\th})
=&\int_{\bt^d}
\frac{e^{\imath\langle {\bf m},{\boldsymbol\th}\rangle}}{\eps+\sum_{j=1}^d(1-\cos\th_j)}
\di m({\boldsymbol\th})\\
=\lim_{\xi\downarrow\eps}&\int_{\bt^d}
\frac{e^{\imath\langle {\bf m},{\boldsymbol\th}\rangle}}{\xi+\sum_{j=1}^d(1-\cos\th_j)}
\di m({\boldsymbol\th})\,.
\end{align*}
In addition, $g_n(\l_n)\to2$ in \eqref{eksek3}, then 
$$
\langle R_{A_{[0,n]}}(g(\l_n))P^\perp_{w_n}\d_k,P^\perp_{w_n}\d_l\rangle\to \langle R_{A_{\bn}}(2)\d_k,\d_l\rangle
$$ 
by Proposition \ref{pefo}, which is finite because $A_{\bn}$ is transient (cf. Proposition \ref{nn11}).
Finally, $\langle P^\perp_{w_n}\d_k,P^\perp_{w_n}\d_l\rangle\to\d_{kl}$ as 
$\langle P_{w_n}\d_k,P_{w_n}\d_l\rangle=w_n(k)w_n(l)\rightarrow0$. 
The proof follows by collecting the previous facts in the formula \eqref{nn0110} giving the resolvent of the comb graph. 
\end{proof}
\begin{Rem} 
The proof of Proposition \ref{cocaszon} suggests some sufficient conditions under which 
$\langle R_{A_{\La_n}}(\l_n)Q_n\d_k\otimes\d_{\bf m},Q_n\d_l\otimes\d_{\bf n}\rangle$ converges also in the case $d\geq3$. Indeed, with the previous notations the sequences
$$
\bigg\{\frac1{\big(\frac{\l_n}2-d\big)(2n+1)^d}\bigg\}_{n\in\bn}\,,\quad \big\{\langle R_{A_{\bt_{2n+1}^d}}(\l_n)\d_{\bf0},\d_{\bf0}\rangle^{-1}\big\}_{n\in\bn}
$$
should converge, and in addition, 
$$
\lim_n\langle R_{A_{\bt_{2n+1}^d}}(\l_n)\d_{\bf0},\d_{\bf0}\rangle^{-1}>\|A_\bn\|\,.
$$
In all these situation, 
$$
\lim_n\langle R_{\La_n}(\l_n)Q_n\d_k\otimes\d_{\bf m},Q_n\d_l\otimes\d_{\bf n}\rangle
\neq\lim_{\l\downarrow\|A_{N_d}\|}\langle R_{A_{N_d}}(\l)\d_k\otimes\d_{\bf m},\d_l\otimes\d_{\bf n}\rangle\,.
$$
We leave the details to the reader.
\end{Rem}
Here, there is the main result describing locally normal states exhbiting BEC as thermodynamic limit of finite volume Gibbs states.
\begin{Thm}
\label{liift010}
For $d=1,2$, let $D\geq0$, and $v$ be the PF weight of $A_{N_d}$ given in \eqref{fifi}, together with the sequence $\{v_n\}_{n\in\bn}$ of the PF eigenvectors for $A_{\La_n}$,  normalised at 1 on the common root $(0,{\bf 0})\in \La_n$. For each sequence of the chemical potentials 
$\{\m_n\}_{n\in\bn}$ with $\m_n<\|A_{N_d}\|-\|A_{\La_n}\|$, such that $\lim_n\m_n=0$ and
$$
\lim_{\m_n\to0}\frac1{\|v_n\|^2(\|A_{N_d}\|-\|A_{\La_n}\|-\m_n)}=D\,,
$$
we get 
\begin{align*}
&\lim_n\left\langle\left(e^{(\|A_{N_d}\|-\m_n)\idd_{\ell^2(\La_n)}-A_{\La_n}}-\idd_{\ell^2(\La_n)}\right)^{-1}
P_{\ell^2(\La_n)}u_1,P_{\ell^2(\La_n)}u_2\right\rangle\\
=&\left\langle\left(e^{\|A_{N_d}\|\idd-A_{N_d}}-\idd\right)^{-1}u_1,u_2\right\rangle+D\langle u_1,v\rangle\langle v,u_2\rangle\quad u_1,u_2\in\gph\,.
\end{align*}
\end{Thm}
\begin{proof}
By Propositions \ref{cocaszon}, \ref{risgfor}, and Theorem \ref{tretre}, the proof follows {\it mutatis--mutandis} the analogous one of Theorem \ref{liift} by using the Analitical Functional Calculus for $e^{\imath t H}$, and Lebesgue Dominated Convergence Theorem. We leave the details to the reader.
\end{proof}

The existence of locally normal states exhibiting BEC is assured because $A_{N_d}$ is transient. It can be constructed as thermodynamic limit by fixing the amount of the condensate as explained in Theorem \ref{liift010}. As $A_{N_d}$ exhibits Hidden Spectrum for $d=1,2$, the critical density is also finite. Thus, it is meaningful to investigate the infinite volume limit by fixing the mean density $\r$ and compute the sequence of the finite volume chemical potential $\m_n$ by solving \eqref{fvden}. The case 
$\r<\r_c$, which corresponds to $\lim_n\m_n<0$, presents no difficulty (cf. \cite{BR2}), so we limit the analysis to the
condensation regime $\r\geq0$ which corresponds to the case $\lim_n\m_n=0$. For this purpose, define for $a\in[0,+\infty]$,
\begin{equation*}
r(a):=\sum_{m=0}^{+\infty}\frac1{a+\pi^2m(m+2)}\,.
\end{equation*}
Such a function is smooth for $a\in(0,+\infty)$, and strictly decreasing with $r(0)=\lim_{a\downarrow0}r(a)=+\infty$, 
$r(\infty)=\lim_{a\to+\infty}r(a)=0$. Put in addition,
\begin{equation}
\label{acalpa}
R(a)=r(a)\,,d=1\,,\quad
R(a)=\frac1a \,,d=2\quad  a\in[0,+\infty]\,.
\end{equation}
\begin{Lemma}
\label{acalpa1}
If $\lim_nn^2\n_n=a$ then 
$$
\lim_n\sum_{m=0}^n\frac{\cos\frac{\pi(m+1)}{n+2}}{n^2\left(\n_n+4\cos\frac{\pi}{n+2}\sin^2\frac{\pi m}{2(n+2)}
+2\sin\frac{\pi}{n+2}\sin\frac{\pi m}{n+2}\right)}=r(a)\,.
$$
\end{Lemma}
\begin{proof}
It is enough to show that
$$
\lim_n\sum_{m=1}^n\frac{\cos\frac{\pi(m+1)}{n+2}}{n^2\left(\n_n+4\cos\frac{\pi}{n+2}\sin^2\frac{\pi m}{2(n+2)}
+2\sin\frac{\pi}{n+2}\sin\frac{\pi m}{n+2}\right)}=
\sum_{m=1}^{+\infty}\frac1{a+\pi^2m(m+2)}\,.
$$
We first notice that
$$
\frac{\chi_{[1,n]}(m)}{n^2\left(\n_n+4\cos\frac{\pi}{n+2}\sin^2\frac{\pi m}{2(n+2)}
+2\sin\frac{\pi}{n+2}\sin\frac{\pi m}{n+2}\right)}
\leq\frac1{2m^2\sin^2\frac{\pi m}{2(m+2)}}\,,
$$
with
$$
\sum_{m=1}^{+\infty}\frac1{2m^2\sin^2\frac{\pi m}{2(m+2)}}<+\infty\,.
$$
The proof follows from Lebesgue Dominated Convergence Theorem.
\end{proof}
We now pass to analyse the behaviour of the infinite volume limit of finite volume two--point function and mean density associated to the Bose--Gibbs grand canonical ensemble. When the mean density is fixed, the corresponding finite volume sequence $\{\m_n\}_{n\in\bn}$ of chemical potential associated to the exhaustion
$\{\La_n\}_{n\in\bn}$ (with $\m_n:=\m(\La_n)$), is determined by solving \eqref{fvden}. As usual we limit the analysis to the condensation regime $\r\geq\r_c$, where for $A_{N_d}$,
$$
\r_c=\frac1{(2\pi)^d}\int_{\bt^d}\frac{\di^d\boldsymbol\th}{e^{(\|A_{N_d}\|-2\sum_{i=1}^d\cos\th_i)}-1}\,,
$$
which leads to $\lim_n\m_n=0$.
\begin{Prop}
\label{fikge}
Let $\{\m_n\}_{n\in\bn}$ be any sequence of chemical potentials such that $\m_n<\|A_{N_d}\|-\|A_{\La_n}\|$ and $\lim_n\m_n=0$.
Suppose further that 
$$
\lim_nn^{d_G(N_d)}(\|A_{N_d}\|-\|A_{\La_n}\|-\m_n)=a\in[0,+\infty]\,.
$$
Then
\begin{align*}
\lim_n\t_n\bigg[&\left(e^{[(\|A_{N_d}\|-\m_n)\idd_{\ell^2(\La_n)}-A_{\La_n}]}-\idd_{\ell^2(\La_n)}\right)^{-1}\bigg]\\
&=\r_c+2^{3-d_G(N_d)}R(b)\|R_{A_{\bz^d}}(\|A_{N_d}\|)\d_{\bf 0}\|^2\,,
\end{align*}
where $R$ is given in \eqref{acalpa}, and
$$
b=\frac{\langle R_{A_{\bz^d}}(\|A_{N_d}\|)^2\d_{\bf 0},\d_{\bf 0}\rangle}
{\langle R_{A_{\bz^d}}(\|A_{N_d}\|)\d_{\bf 0},\d_{\bf 0}\rangle^2}a\,.
$$
\end{Prop}
\begin{proof}
By taking into account the form of the Resolvent in Proposition \ref{risgfor}, and reasoning as in Proposition 5.6 of \cite{F2}, we get with $g_n$ given in \eqref{gnegne},
\begin{align*}
\r_{\La_n}(\m_n)-\r_c=&\D_n
+\frac{g_n(\l_n)}{(2n+1)^d(n+1)}\sum_{|k_j|\leq n}|\langle R_{A_{\bt^d_{2n+1}}}(\l_n)\d_{\bf k},\d_{\bf 0}\rangle|^2\\
\times&\sum_{m=0}^n\frac{2\cos\frac{\pi(m+1)}{n+2}}{\n_n+4\cos\frac{\pi}{n+2}\sin^2\frac{\pi m}{2(n+2)}
+2\sin\frac{\pi}{n+2}\sin\frac{\pi m}{n+2}}\,.
\end{align*}
Here, $\n_n:=g_n(\l_n)-2\cos\frac{\pi}{n+2}$, $g_n(\l_n)\to2$, and for the continuous function $f$ given in \eqref{resam} with obvious notations,
\begin{align*}
\D_n:=&\t^{N_d}_n\left[f\big((\|A_{N_d}\|-\m_n)\idd_n-A_{\La_n}\big)\right]-\t^{N_d}\left[f\big(\|A_{N_d}\|\idd-A_{N_d}\big)\right]\\
+&\t^{\bz^d}_n\big(R_{A_{\bt^d_{2n+1}}}(\l_n)\big)-\t^{\bz^d}\big(R_{A_{\bz^d}}(\|A_{N_d}\|)\big)
\end{align*}
goes to $0$ as $n\to+\infty$. By using Taylor expansion of the function $g_n$ in the right neighbourhood of $\|A_{S_n}\|$, we get with 
$\l_n=\|A_{N_1}\|-\m_n$,
$$
g_n(\l_n)-\|A_{S_n}\|=g'(\s_n)(\l_n-\|A_{\La_n}\|)\,,
$$
with $\s_n\in(\|A_{\La_n}\|,\l_n)$. Thus, $n^l(\l_n-\|A_{\La_n}\|)\to a$ implies $n^l\n_n\to b$ with $b$ as above. Now, if $d=1$ the proof directly follows by Lemma \ref{acalpa1}. If $d=2$, in the sum
$$
\frac1n\sum_{m=0}^n\frac1{n^2\left(\n_n+4\cos\frac{\pi}{n+2}\sin^2\frac{\pi m}{2(n+2)}
+2\sin\frac{\pi}{n+2}\sin\frac{\pi m}{n+2}\right)}
$$
only the term corresponding to $m=0$ survives in the limit $n\to\infty$, thanks again to Lemma \ref{acalpa1}, and the proof follows withf $R$ given in \eqref{acalpa}.
\end{proof}
The following result explains the differences between the two alternatives 
$d=1$ which leads to some unexpected effects due to $d_{PF}>d_G$, and $d=2$ where $d_{PF}=d_G$ where the emerging results parallel the known ones for homogeneous systems. We formulate the following result in a way to encode also those for which $d_{PF}<d_G$ occurring in the forthcoming section, provided that the Adjacency is transient and the critical density is finite, see Remark \ref{frfinr}.
Indeed, let $\a(x)$ be the extended--valued function defined as 
\begin{equation}
\label{aficas}
\a(x)= 
     \begin{cases}
     0\,,&x<0\,,\\
1\,,&x=0\,,\\
+\infty\,,&x>0\,.
     \end{cases}
\end{equation}
\begin{Prop}
\label{ifcavozt}
Fix $\r\geq\r_c$. With the sequence of chemical potential $\{\m_n\}_{n\in\bn}$ obtained by solving \eqref{fvden}, we get for
$u_1,u_2\in\gph$,
\begin{align}
\label{liift0110ca}
&\lim_n\left\langle\left(e^{(\|A_{N_d}\|-\m_n)\idd_{\ell^2(\La_n)}-A_{\La_n}}-\idd_{\ell^2(\La_n)}\right)^{-1}
P_{\ell^2(\La_n)}u_1,P_{\ell^2(\La_n)}u_2\right\rangle\nn\\
=&\left\langle\left(e^{\|A_{N_d}\|\idd-A_{N_d}}-\idd\right)^{-1}u_1,u_2\right\rangle\\
+&
\frac{2\pi^2(\r-\r_c)\a\big(d_G(N_d)-d_{PF}(N_d)\big)}
{\langle R_{A_{\bz^d}}(\|A_{N_d}\|)\d_{\bf 0},\d_{\bf 0}\rangle^2
\sum_{{\bf k}\in\bz^d}|\langle R_{A_{\bz^d}}(\|A_{N_d}\|)\d_{\bf k},\d_{\bf 0}\rangle|^2}
\langle u_1,v\rangle\langle v,u_2\rangle\nn\,.
\end{align}
\end{Prop}
\begin{proof}
We start by recalling that for the case $\r\geq\r_c$, necessarily $\lim_n\m_n=0$, see e.g. 
\cite{BR2, BCRSV, F, F2, FGI1}. We start by considering each converging subsequence 
$\big\{n_k^{d_G(N_d)}(\|A_{N_d}\|-\|A_{\La_{n_k}}\|-\m_{n_k})\big\}_{k\in\bn}$. By Proposition \ref{fikge}, we necessarily get 
$$
\lim_kn_k^{d_G(N_d)}(\|A_{N_d}\|-\|A_{\La_{n_k}}\|-\m_{n_k})=a\,,
$$ 
where $a$ is the unique solution of the equation
$$
\r-\r_c=2^{3-d_G(N_d)}R\left(a\frac{\langle R_{A_{\bz^d}}(\|A_{N_d}\|)^2\d_{\bf0},\d_{\bf0}\rangle}{\langle R_{A_{\bz^d}}(\|A_{N_d}\|)\d_{\bf0},\d_{\bf0}\rangle^2}\right)\sum_{{\bf k}\in\bz^d}|\langle R_{A_{\bz^d}}(\|A_{N_d}\|)\d_{\bf k},\d_{\bf0}\rangle|^2\,.
$$
Thus, $\lim_nn^{d_G(N_d)}(\|A_{N_d}\|-\|A_{\La_{n}}\|-\m_{n})=a\in(0,+\infty]$ with $R$ given in \eqref{acalpa},
where the case $+\infty$ corresponds to $\r=\r_c$. If $d=2$, thanks to Propositions \ref{pfene}, \ref{pfenenz}, and \eqref{fifi}, the latter corresponds to $D=0$ in Theorem \ref{liift010}, and the remaining ones to $D>0$ with 
$$
D=\frac{2\pi^2(\r-\r_c)}
{\langle R_{A_{\bz^d}}(\|A_{N_d}\|)\d_{\bf 0},\d_{\bf 0}\rangle^2
\sum_{{\bf k}\in\bz^d}|\langle R_{A_{\bz^d}}(\|A_{N_d}\|)\d_{\bf k},\d_{\bf 0}\rangle|^2}\,.
$$
If $d=1$, in all these cases
$\lim_nn^3(\|A\|-\|A_{n}\|-\m_{n})=+\infty$, which correspond to $D=0$ in Theorem \ref{liift010} thanks again to Proposition \ref{pfenenz}.
\end{proof}
\begin{Rem}${}$\\
{\bf(i)} Concerning the infinite volume limit of finite volume Gibbs states in Theorem \ref{liift010}, in order to construct states exhibiting a condensate density $D>0$, analogous considerations as those in Remark \ref{reenai} can be done with the obvious modification of (iv). Indeed, for the infinite volume behaviour of the finite volume condensate when $D=0$ we have
\begin{align*}
\r^{{\text cond}}&\big(\om_{\La_n}\big)=\frac{\|v_n\|^2}{|\La_n|}\frac1{\|v_n\|^2(\|A\|-\|A_n\|-\m_n)}\\
=&\bigg(\frac1{2\pi^2}+o(1)\bigg)\frac{n^{2-d}}{\|v_n\|^2(\|A\|-\|A_n\|-\m_n)}\,,
\end{align*}
whose possible limit depends for $d=1$, on the rate of the convergence to 0 of $\frac1{\|v_n\|^2(\|A\|-\|A_n\|-\m_n)}$.\\
\vskip.1cm
\noindent
{\bf(ii)} If conversely we fix the density $\r>\r_c$, we must distinguish the case $d=1$ where $d_{PF}>d_G$, from $d=2$ where $d_{PF}=d_G$. If $d=1$ we always have 
$\r(\om)=\r_c$ (hence $\r_{\La_n}(\om)=\r_c$ for all the limiting local densities of the condensate), even if we have the constrain 
$\r(\om_{\La_n})=\r>\r_c$. This is the simple consequence of the fact that 
$d_{PF}>d_G$.
\vskip.1cm
\noindent
{\bf(iii)} The case $d=2$ behaves like transient homogeneous lattices because $d_{PF}=d_G$. In this case, 
$\r^{cond}(\om)\propto(\r-\r_c)$, but the proportionality constant is different from 1. This can happen already in the homogeneous known cases, see e.g. Theorem 5.2.32 of \cite{BR2}.
\end{Rem}

\section{The comb graph $\bz^d\comb\bz$}
\label{sec:ze}

Another case of interest for our purposes for which all the calculations can be carried out, is the comb graphs $C_d:=\bz^d\comb\bz$. In this situation, we use the periodic boundary condition for the Adjacency of the finite volume theories on both base and fiber space, without affecting the substance of the analysis. Namely, $\La_n:=\bt^d_{2n+1}\comb\bt_{2n+1}$, and the finite volume Adjacencies $A_{\La_n}$ again provide an additive negligible perturbation of $\bz^d$--copies of $\bz$, see Proposition \ref{pradneg}. Thanks to 
$$
\|A_{C_d}\|=2\sqrt{d^2+1}>2=\|A_{\bz}\|\,,
$$
all the networks $C_d$ admits Hidden Spectrum (for the adjacency matrix), then the critical density of the Pure Hopping model is always finite. In addition, $A_{C_d}$ is transient if and only if $d\geq3$, then the investigation of the BEC is meaningful only in this situation, see Proposition \ref{roec}.
Finally,
$$
d=d_{PF}(C_d)<d_G(C_d)=d+1\,,
$$
then we will be in the opposite situation of $N_1$ of Section \ref{snd}. We also refer the reader to Sections 9 and 10 of \cite{FGI1} containing some previous results about 
the Pure Hopping model on the comb graphs $C_d$.

Consider the PF eigenvector $v_n$ of $A_{\La_n}$, normalised at 1 on the common root $({\bf0},0)$. It converges point--wise to the PF weight on $C_d$ given by 
\begin{equation*}
v({\bf k},m)=\frac{d}{2\pi}\int_{\bt}\frac{e^{\imath m\th}}{\sqrt{d^2+1}-\cos\th}
\di\th\,,\quad ({\bf k},m)\in\bz^d\times\bz\,.
\end{equation*} 
As in the previous section, we pass to the investigation of the infinite volume limit of the finite volume approximations. We start with the two--point function, and for such a purpose consider the self--adjoint projection $Q_n=P_{v_n}^\perp$ onto the codimension $1$ orthogonal subspace to $v_n$ in $\ell^2(\La_n)$. For $d\geq3$, we can construct locally normal states exhibiting BEC by infinite volume limits of Bose--Gibbs states, as we are going to see. 
\begin{Prop}
\label{cocaszon00}
Consider each sequence $\{\l_n\}_{n\in\bn}$ with $\l_n>\|A_{\La_n}\|$, $n\in\bn$, such that
$\lim_n\l_n=\|A_{C_d}\|$. The following assertions hold true. If $d\geq3$, then 
\begin{align*}
&\lim_n\langle R_{A_{\La_n}}(\l_n)Q_n\d_{\bf k}\otimes\d_m,Q_n\d_{\bf l}\otimes\d_j\rangle\\
=&\frac{\prod_{i=1}^d\d_{k_il_i}}{4\pi}\int_{-\pi}^{\pi}
\frac{e^{\imath(m-j)\th}}{\sqrt{d^2+1}-\cos\th}
\di\th
+\frac{d}{2(2\pi)^{d+2}}\int_{\bt^d}
\frac{e^{\imath\langle {\bf k}-{\bf l},{\boldsymbol\th}\rangle}\sum_{j=1}^d\cos\th_j}{d-\sum_{j=1}^d-\cos\th_j}
\di^d{\boldsymbol\th}\\
\times&
\int_{-\pi}^{\pi}\!\int_{-\pi}^{\pi}
\frac{e^{\imath(m\a-j\b)}}
{(\sqrt{d^2+1}-\cos\a)(\sqrt{d^2+1}-\cos\b)}
\di\a\di\b\\
=&\lim_{\l\downarrow\|A_{C_d}\|}\langle R_{A_{C_d}}(\l)\d_{\bf k}\otimes\d_m,\d_{\bf l}\otimes\d_j\rangle\,.
\end{align*} 
If $d=1,2$ then
\begin{align*}
\lim_n\langle R_{A_{\La_n}}&(\l_n)Q_n\d_{\bf k}\otimes\d_m,Q_n\d_{\bf k}\otimes\d_m\rangle
=+\infty\\
=\lim_{\l\downarrow\|A_{C_d}\|}&\langle R_{A_{C_d}}(\l)\d_{\bf k}\otimes\d_m,\d_{\bf k}\otimes\d_m\rangle\,.
\end{align*} 
\end{Prop}
\begin{proof}
The proof follows as that of Proposition \ref{cocaszon} by taking into account that $v_n=w_n\otimes r_n$, with constant $w_n$ given by
$w_n({\bf k})=\frac1{(2n+1)^{d/2}}$ such that $\|w_n\|=1$. In fact, as before,
\begin{align*}
Q_nR_{A_{\La_n}}(\l_n)Q_n=&P_{w_n}\otimes(\idd-P_{r_n})R_{A_{\La_n}}(\l_n)P_{w_n}\otimes(\idd-P_{r_n})\\
+&P^\perp_{w_n}\otimes \idd R_{A_{\La_n}}(\l_n)P^\perp_{w_n}\otimes \idd\,,
\end{align*}
and the matrix elements of the first addendum go to zero in all the situation because $\|A_{C_d}\|>\|A_{\bz^d}\|$. Concerning the 
second one, its matrix elements converge provided $A_{\bz^d}$ is transient, that is when $d\geq3$. Conversely, with $g_n$ given in \eqref{gnegne} and $\n_n=\frac{g_n(\l_n)}2-d$,
if $A_{\bz^d}$ is recurrent which corresponds to $d=1,2$, the diagonal part of the matrix elements of the second addendum contains the factor
$$
\langle P^\perp_{w_n}R_{A_{\bt^d_{2n+1}}}(g_n(\l_n))A_{\bt^d_{2n+1}}P^\perp_{w_n}\d_{\bf k},\d_{\bf k}\rangle
=\frac12\int_{\bt^d_{2n+1}\backslash\{{\bf 0}\}}\frac{\sum_{j=1}^d\cos\th_j}{\n_n+\sum_{j=1}^d(1-\cos\th_j)}\di m_n({\boldsymbol\th})
$$
which diverges because $\n_n\to0$.
\end{proof}
As $A_{C_d}$ is recurrent if $d=1,2$, the Pure Hopping model on the network $C_d$ cannot exhibit BEC at all, see Proposition \ref{roec}. Thus, we study the thermodynamic limit of the two--point function after fixing the amount of the condensate, in the transient case
\begin{Thm}
\label{liift010czac}
Let $D\geq0$ and $d\geq3$. For each sequence of the chemical potential 
$\{\m_n\}_{n\in\bn}$ with $\m_n<\|A_{C_d}\|-\|A_{\La_n}\|$, such that $lim_n\m_n=0$ and
$$
\lim_{\m_n\to0}\frac1{\|v_n\|^2(\|A_{C_d}\|-\|A_{\La_n}\|-\m_n)}=D\,,
$$
we get 
\begin{align*}
&\lim_n\left\langle\left(e^{(\|A_{C_d}\|-\m_n)\idd_{\ell^2(\La_n)}-A_{\La_n}}-\idd_{\ell^2(\La_n)}\right)^{-1}
P_{\ell^2(\La_n)}u_1,P_{\ell^2(\La_n)}u_2\right\rangle\\
=&\left\langle\left(e^{\|A_{C_d}\|\idd-A_{C_d}}-\idd\right)^{-1}u_1,u_2\right\rangle+D\langle u_1,v\rangle\langle v,u_2\rangle\quad u_1,u_2\in\gph\,.
\end{align*}
\end{Thm}
\begin{proof}
The proof follows as the previous results in Theorems \ref{liift}, \ref{liift010}, by taking into account that 
$$
v_n({\bf k},j)\leq1\,,\quad ({\bf k},j)\in\bz^d\times\bz\,,
$$
\begin{align*}
&|\langle R_{A_{\bt_{2n+1}^d}}(\l_n)Q_n\d_{\bf k},Q_n\d_{\bf l}\rangle|
=\int_{\bt_{2n+1}^d\backslash\{{\bf0}\}}\frac{\di m_n({\boldsymbol\th})}{\l_n-2\sum_{i=1}^d\cos\th_i}\\
\leq&\int_{\bt_{2n+1}^d\backslash\{{\bf0}\}}\frac{\di m_n({\boldsymbol\th})}{2d-2\sum_{i=1}^d\cos\th_i}
\to\int_{\bt^d}\frac{\di m({\boldsymbol\th})}{2d-2\sum_{i=1}^d\cos\th_i}\\
=&\langle R_{A_{\bz^d}}(2d)\d_{\bf 0},\d_{\bf 0}\rangle\,,
\end{align*}
that is $\langle R_{A_{\bt_{2n+1}^d}}(\l_n)Q_n\d_{\bf k},Q_n\d_{\bf l}\rangle$ is uniformly bounded. We leave the details to the reader.
\end{proof}
We now pass to the thermodynamic limit by fixing the mean density of the quasi--free state under consideration, instead of the amount of the condensate. For $a\in[0,+\infty]$ define
\begin{equation*}
s(a):=\sum_{m=0}^{+\infty}\frac1{a+2\pi^2m^2}\,.
\end{equation*}
Such a function is smooth for $a\in(0,+\infty)$, and strictly decreasing with $s(0)=\lim_{a\downarrow0}s(a)=+\infty$, 
$s(\infty)=\lim_{a\to+\infty}s(a)=0$. Put in addition,
\begin{equation}
\label{acalpafca}
S(a)=s(a)\,,\,\,d=1\,,\quad
S(a)=\frac1a \,,\,\,d\geq2\,.
\end{equation}
\begin{Lemma}
\label{abalpa1}
If $\lim_n(2n+1)^2\n_n=a$ then 
$$
\lim_n\sum_{m=0}^n\frac{\cos\frac{2\pi m}{2n+1}}{(2n+1)^2\n_n+(2n+1)^2\left(1-\cos\frac{2\pi m}{2n+1}\right)}=s(a)\,.
$$
\end{Lemma}
\begin{proof}
The proof follows from Lebesgue Dominated Convergence Theorem by noticing that
$$
\frac{\chi_{[1,n]}(m)}{(2n+1)^2\n_n+(2n+1)^2\left(1-\cos\frac{2\pi m}{2n+1}\right)}
\leq\frac1{(2m+1)^2\left(1-\cos\frac{2\pi m}{2m+1}\right)}\leq\frac1{6m^2}\,,
$$
with $\sum_{m=1}^{+\infty}\frac1{6m^2}<+\infty$.
\end{proof}
As $\r_c$ is always finite, we see that the finite--volume densities always converges even if we fix the mean density of the finite--volume states $\r\geq\r_c$ for all the networks $C_d$, included the recurrent cases $d=1,2$.
Also for all the cases under consideration in the present section, we report the formula for the critical density 
of the Pure Hopping model on $C_d$:
$$
\r_c=\frac1{2\pi}\int_{-\pi}^\pi\frac{\di\th}{e^{2(\sqrt{d^2+1}-\cos\th)}-1}\,,
$$
see Proposition 9.2 of \cite{FGI1}. 
\begin{Prop}
\label{fikgeca}
Let $\{\m_n\}_{n\in\bn}$ be any sequence of chemical potentials such that $\m_n<\|A_{N_d}\|-\|A_{\La_n}\|$ and $\lim_n\m_n=0$.
Suppose further that 
$$
\lim_nn^{d_G(C_d)}(\|A_{C_d}\|-\|A_{\La_n}\|-\m_n)=a\in[0,+\infty]\,.
$$
Then
\begin{align*}
\lim_n\t_n\bigg[&\left(e^{[(\|A_{C_d}\|-\m_n)\idd_{\ell^2(\La_n)}-A_{\La_n}]}-\idd_{\ell^2(\La_n)}\right)^{-1}\bigg]\\
&=\r_c+2d^2S(b)\|R_{A_{\bz}}(\|A_{C_d}\|)\d_0\|^2\,,
\end{align*}
where $S$ is given in \eqref{acalpafca}, and
$$
b=2^d\frac{\langle R_{A_{\bz}}(\|A_{C_d}\|)^2\d_{0},\d_{0}\rangle}
{\langle R_{A_{\bz}}(\|A_{C_d}\|)\d_{0},\d_{0}\rangle^2}a\,.
$$
\end{Prop}
\begin{proof}
The proof follows the same lines of the analogous one of Proposition \ref{fikge}. Indeed, by taking into account of the form of the Resolvent in Proposition \ref{risgfor}, and reasoning as in Proposition 5.6 of \cite{F2}, we get
\begin{align*}
\r_{\La_n}(\m_n)-\r_c=&\D_n
+g_n(\l_n)\sum_{|k|\leq n}|\langle R_{A_{\bt_{2n+1}}}(\l_n)\d_k,\d_0\rangle|^2\\
\times&\frac1{2n+1}\int_{\bt^d_{2n+1}}\frac{\sum_{j=1}^d\cos\th_j}{\n_n+\sum_{j=1}^d(1-\cos\th_j)}\di m_n(\th)
\end{align*}
Here, $\n_n:=g_n(\l_n)/2-d$, $g_n(\l_n)\to2d$, and finally for the continuous function $f$ given in \eqref{resam},
\begin{align*}
\D_n:=&\t^{C_d}_n\left[f\big((\|A_{C_d}\|-\m_n)\idd_n-A_{\La_n}\big)\right]-\t^{C_d}\left[f\big(\|A_{C_d}\|\idd-A_{C_d}\big)\right]\\
+&\t^{\bz}_n\big(R_{A_{\bt_{2n+1}}}(\l_n)\big)-\t^{\bz}\big(R_{A_\bz}(2\sqrt{d^2+1})\big)
\end{align*}
goes to $0$ as $n\to+\infty$. For $d=1$ we get
$$
\int_{\bt_{2n+1}}\frac{\cos\th\di m_1(\th)}{\n_n+(1-\cos\th)}
=\sum_{m=0}^n\frac{\cos\frac{2\pi m}{2n+1}}{(2n+1)\n_n+(2n+1)\left(1-\cos\frac{2\pi m}{2n+1}\right)}\,,
$$
and the proof follows as in Proposition \ref{fikge} by taking into account Lemma \ref{abalpa1}. The case $d>1$ leads to
$$
\frac1{2n+1}\int_{\bt^d_{2n+1}\backslash\{0\}}\frac{\di m_n(\th)}{\sum_{j=1}^d(1-\cos\th_j)}
\approx\frac1n\int_{1/n}^1x^{d-3}\di x\to0\,.
$$
Thus, the unique surviving term is that localised on the origin of $\bt^d_{2n+1}$ and the proof follows as before.
\end{proof}
Now we end with the thermodynamic limit when the mean density of the state is fixed. As usual, we limit the analysis to the condensation regime $\r\geq\r_c$ by putting
$\m_n:=\m(\La_n)$ for the finite volume chemical potential.
\begin{Prop}
\label{ifcaiizt}
Fix $\r>\r_c$. With the sequence of chemical potential $\{\m_n\}$ obtained by solving \eqref{fvden}, we get for
$u\in\gph$,
\begin{equation}
\label{trrec}
\lim_n\left\langle\left(e^{(\|A_{C_d}\|-\m_n)\idd_{\ell^2(\La_n)}-A_{\La_n}}-\idd_{\ell^2(\La_n)}\right)^{-1}
P_{\ell^2(\La_n)}u,P_{\ell^2(\La_n)}u\right\rangle=+\infty\,.
\end{equation}
\end{Prop}
\begin{proof}
In the condensation regime, we necessarily have $\m_n\to0$. By reasoning as in Proposition \ref{ifcavozt}, we get 
$\lim_nn^{d_G(C_d)}(\|A_{C_d}\|-\|A_{\La_{n}}\|-\m_{n})=a\in(0,+\infty)$, where $a$ is the unique solution of 
$$
\r-\r_c=2d^2S\bigg(2^d\frac{\langle R_{A_{\bz^d}}(\|A_{C_d}\|)^2\d_{\bf 0},\d_{\bf 0}\rangle}
{\langle R_{A_{\bz^d}}(\|A_{C_d}\|)\d_{\bf 0},\d_{\bf 0}\rangle^2}a\bigg)\|R_{A_{\bz}}(\|A_{C_d}\|)\d_0\|^2\,.
$$
Now with $v_n$ the finite volume PF eigenvector normalised at 1 on $({\bf0},0)$,
\begin{align*}
&\left\langle\left(e^{(\|A_{C_d}\|-\m_n)\idd_{\ell^2(\La_n)}-A_{\La_n}}-\idd_{\ell^2(\La_n)}\right)^{-1}
P_{\ell^2(\La_n)}u,P_{\ell^2(\La_n)}u\right\rangle\\
\geq&\frac{|\langle u,v_n\rangle|^2}
{\|v_n\|^2(\|A_{C_d}\|-\|A_{\La_n}\|-\m_n)}
\approx\frac1{n^{d_{PF}(C_d)}(\|A_{C_d}\|-\|A_{\La_{n}}\|-\m_{n})}\\
=&\frac{n}{n^{d_{G}(C_d)}(\|A_{C_d}\|-\|A_{\La_{n}}\|-\m_{n})}\to+\infty\,.
\end{align*}
\end{proof}
\begin{Rem}
\label{frfinr}
If $d=1,2$ and $\r=\r_c$, \eqref{trrec} holds true because $C_d$ is recurrent, see Proposition \ref{roec}. Conversely, for the case 
$d\geq3$ and $\r=\r_c$, we have by Proposition \ref{ifcavozt},
$$
\lim_nn^{d_G(C_d)}(\|A_{C_d}\|-\|A_{\La_{n}}\|-\m_{n})=+\infty\,,
$$
and a more careful analysis is needed to study the infinite volume behaviour of

$\left\langle\left(e^{(\|A_{C_d}\|-\m_n)\idd_{\ell^2(\La_n)}-A_{\La_n}}-\idd_{\ell^2(\La_n)}\right)^{-1}
P_{\ell^2(\La_n)}u,P_{\ell^2(\La_n)}u\right\rangle$. This again explains that for such inhomogeneous systems, in order to construct locally normal states exhibiting BEC, it is more natural to fix the amount of the condensate instead of the mean density. As $d_G>d_{PF}$, if $d\geq3$ and $\r>\r_c$, \eqref{trrec} in Proposition \ref{ifcaiizt} can be still expressed in the form similar to that of \eqref{liift0110ca} by using the function $\a$ in \eqref{aficas}.
\end{Rem}

\section*{Acknowledgement} 
The author kindly acknowledges Yuri Safarov for some useful suggestions relative to Proposition \ref{astc}.

\end{document}